\documentclass[a4paper,11pt]{article}      
\usepackage{graphicx,pstricks}
\usepackage[round]{natbib}
\usepackage[english]{babel}
\usepackage{amssymb,amsmath,amsfonts,amsthm}

\usepackage{dcolumn}					
\usepackage[onehalfspacing]{setspace}			
\usepackage[top=3cm, bottom=3cm]{geometry}		
\usepackage{eurosym}					
\usepackage{enumerate}   
\usepackage{xcolor,soul} 

\newcommand{\RR}{\mathbb{R}}

\newcommand{\set}[2]{\left\{#1\;\left|\; \; #2\right.\right\}}

\definecolor{lightgray}{gray}{0.80}
\sethlcolor{lightgray}

\newtheorem{theorem}{Theorem}
\newtheorem{corollary}{Corollary}

\newtheorem{proposition}{Proposition}
\theoremstyle{remark}

\theoremstyle{definition}
\newtheorem{definition}{Definition}
\newtheorem{example}{Example}
 
\begin{document}
 
 
\title{Time Will Tell:\\ Recovering Preferences when Choices Are Noisy\thanks{We thank Drew Fudenberg, Ian Krajbich, Paulo Natenzon, Antonio Rangel, Jakub Steiner, Tomasz Strzalecki, Ryan Webb, and Michael Woodford for helpful comments. We are also grateful to seminar audiences in the Workshop on Neuroeconomics and the Evolution of Economic Behavior (Vancouver, 2018), the workshop on Stochastic Choice (Barcelona, 2018), and the Sloan-Nomis Workshop on the Cognitive Foundations of Economic Behavior (Vitznau, Switzerland, 2018).}}


%

\author{ 
\begin{tabular}{ccc}
Carlos Al\'os-Ferrer\thanks{Email: {carlos.alos-ferrer@econ.uzh.ch}. Department of Economics, University of Zurich. Bl\"umlisalpstrasse 10, 8006 Zurich, Switzerland.} & Ernst Fehr\thanks{Email: {ernst.fehr@econ.uzh.ch}. Department of Economics, University of Zurich. Bl\"umlisalpstrasse 10, 8006 Zurich, Switzerland.}  & Nick Netzer\thanks{Email: {nick.netzer@econ.uzh.ch}. Department of Economics, University of Zurich. Bl\"umlisalpstrasse 10, 8006 Zurich, Switzerland.}\\ 
\multicolumn{3}{c}{Department of Economics, University of Zurich} 
\end{tabular}
}  
 

 \date{This Version: October 2018}

 
\maketitle
\thispagestyle{empty}


\begin{abstract} 
The ability to uncover preferences from choices is fundamental for both positive economics and welfare analysis. Overwhelming evidence shows that choice is stochastic, which has given rise to random utility models as the dominant paradigm in applied microeconomics. However, as is well known, it is not possible to infer the structure of preferences in the absence of assumptions on the structure of noise. This makes it impossible to empirically test the structure of noise independently from the structure of preferences. Here, we show that the difficulty can be bypassed if data sets are enlarged to include response times. A simple condition on response time distributions (a weaker version of first-order stochastic dominance) ensures that choices reveal preferences without assumptions on the structure of utility noise. Sharper results are obtained if the analysis is restricted to specific classes of models. Under symmetric noise, response times allow to uncover preferences for choice pairs outside the data set, and if noise is Fechnerian, even choice probabilities can be forecast out of sample. We conclude by showing that standard random utility models from economics and standard drift-diffusion models from psychology necessarily generate data sets fulfilling our sufficient condition on response time distributions.

\vspace{0.3cm}



\end{abstract}


\vfill
\noindent {\bf Working Paper.} This is an author-generated version of a research manuscript which is circulated exclusively for the purpose of facilitating scientific discussion. All rights reserved. The final version of the article might differ from this one.

 \newpage 
%
%
%

%
%


\setcounter{page}{1}

\section{Introduction}

Revealed-preference arguments lie at the foundation of economics \citep[e.g.][]{Samuelson38,Houthakker50,Arrow59}. Preferences revealed by choice are used in positive economics to predict behavior in novel situations, and in normative economics to evaluate the desirability of economic policies. In fact, the very use of a utility function entails the implicit assumption that the preferences represented by that function have been deduced from observations of choice behavior. 

The traditional revealed-preference approach assumes that choices are deterministic. This assumption is contradicted by real-world choice behavior, as argued already in the classical work of \citet{Fechner1860} and \citet{Luce59}. Extensive evidence shows that individuals often make different choices when confronted with the same set of options repeatedly
\citep[among many others, see][]{Tversky69,Camerer89,HeyOrme94,AgranovOrtoleva17}. In view of this evidence, the traditional approach has been modified by adding a random component to (cardinal) utility \citep[e.g.][]{Thurstone27,Marschak60,McFadden01}. Random utility has several different interpretations, ranging from noise in an individual's perception of the options, over temporary fluctuations of tastes, to unobserved heterogeneity in a population of agents. With assumptions about the distribution of the random utility component, it becomes possible to deduce an underlying deterministic utility function from observed choice behavior. Depending on the interpretation, this utility function represents either the true preferences of one given individual, or the average preferences in a population.

A major problem with the random utility approach is that the distributional assumptions may actually drive the results. It is a well-known (but rarely stated) fact in stochastic choice theory that nothing can be learned about preferences without making distributional assumptions. The flipside of this result is that anything can be ``learned'' by making the suitable assumptions on the structure of noise. Unfortunately, these assumptions cannot be verified, because utility is a latent variable that is not directly observed. Hence, the dependence on possibly unwarranted assumptions is not just an abstract theoretical problem but is known to plague empirical research \citep[see][]{HeyOrme94}. For instance, \citet{BuschenaZilberman00} showed that, for the same data set, assuming homoskedasticity supports non-expected utility models, while expected utility models cannot be improved upon when heteroskedasticity is allowed.

In this paper, we show that the problem can be overcome by using data on response times. This is because the distribution of response times, which is in principle observable, contains information about the unobservable distribution of utility.\footnote{The time it takes to make a decision (``response time'') can always be observed in laboratory experiments, but even outside the laboratory response times are in principle an observable outcome of a choice process that can be collected by researchers, firms, or other interested parties. In contrast, the distribution of utility noise is intrinsically unobservable.} We derive, first, a simple and intuitive condition on the distribution of response times that ensures that preferences can be identified from choice data without any assumptions on the structure of noise. Second, we show that under symmetric noise, response times enable the identification of preferences between alternatives for which no previous choice data exist. This would not be possible without response time data unless one is willing to impose much stronger (untestable) assumptions than symmetry on the utility noise. Third, we show that if one is willing to assume that utility noise is Fechnerian---an underlying assumption of probit and logit models---response time data enable the calculation of precise choice probabilites for alternatives for which no choice data exist. Again, this would not be possible without response time data. 

Our approach is made possible by the fact that, despite being stochastic, choice behavior obeys certain well-known regularities. One of those regularities, often referred to as the {\em psychometric function}, is the fact that easier choice problems are more likely to elicit correct responses than harder problems. This can be traced back to perceptual discrimination experiments in psychophysics, where an objectively correct response exists (e.g.\ choosing the brightest or loudest stimulus). It is perhaps one of the most robust facts in all of psychology that the percentage of correct choices increases with the difference in stimuli \citep{Laming85,Klein01,WichmannHill01}. Conversely, choice becomes noisier when stimuli are more similar and hence the problem is harder. This finding also extends to cases where the correct response is subjective (e.g. favorite colors) and is uncovered by the researcher through ratings \citep{Dashiell37}. In economics, the classical work of \citet{MostellerNogee51} showed that the phenomenon also occurs in decisions under risk, with Bernoulli utilities estimated by assuming that choice frequencies close to $1/2$ reveal indifference and using linear interpolation to pin down other utility values.\footnote{Specifically, \citet{MostellerNogee51} used bet/pass decisions with choice frequencies close to $1/2$ to write down indifference equations, which allowed them to derive utility values for a subset of monetary values and then to extend them to the whole space by piecewise linear interpolation. Since the remaining decisions were not used for the estimation procedure, they could be used to investigate the relation between choice frequencies and utility differences.} In their data, alternatives with larger estimated utility were not always chosen, but the percentage of choices in favor of the high-utility option increased in the utility difference between the options. This corresponds exactly to the psychometric function, with the difficulty of the (binary) choice problem being measured by the subjective utility difference between the available options, and easier choices being those with a larger absolute utility difference. In fact, this psychometric relationship is an integral part of standard random utility models, which assume that choice probabilities are monotone in utility differences. 

Our approach in this paper rests on integrating a second well-known regularity, often referred to as the {\em chronometric function}, into the standard random utility framework. The chronometric function describes the fact that easier choice problems take less time to respond than harder problems. As in the case of the psychometric function, there is overwhelming evidence from psychophysics showing this regularity. For instance, \citet{MoyerLandauer67} demonstrate that chronometric effects exist even for the simple question of which of two numbers is larger \citep[see also][]{MoyerBayer76,DehaeneDupouxMehler90}. The finding extends to choice based on subjective preferences, as in the work on favorite colors by \citet{Dashiell37}. \citet{Chabrisetal09} and \citet{CardsRT} show that the phenomenon also occurs in intertemporal decisions and decisions under risk, respectively, with utilities estimated following a logit specification for the distribution of noise. That is, response times are decreasing in estimated utility differences.\footnote{As observed by \citet{KrajbichOudFehr14}, the chronometric effect is paradoxical from an economics point of view because it implies that more time is spent on decisions where the stakes are lower. An extreme example is the well-known ``Buridan's ass'' paradox, where a decision-maker takes forever to make a decision among identically evaluated options. The more appropriate interpretation of the phenomenon is that it reflects neural processes which uncover the difference in values between the options, and it takes longer to differentiate closer values. In other words, it takes a longer time to discover that one is indifferent than to recognize a clear preference.} While the chronometric function is just the twin of the psychometric function, and the latter is already incorporated into standard economics models, the literature has only recently started to apply the chronometric relationship in economic models of choice \citep[e.g.][]{Krajbichetal15,PrefRev,FudenbergStrackStrzalecki18}.\footnote{Recent empirical work \citep{SchotterTrevino14,KonovalovKrajbich17,Clithero18} has concentrated on the use of response times in structural models with specific assumptions on error distributions. We will discuss the relation of our work to these contributions in Section \ref{Subsec.LitRT}.} Presumably this is because classical choice theory has not been interested in response times. We are also not interested in response times {\em per se}, but we demonstrate the value of using them in the classical revealed preference approach. That is, we view response times strictly as a tool to make better inferences about the structure of preferences.

Like much of the related literature, we will focus on binary choice problems in this paper. We will assume that, in any given instance of the problem, the realized response time is a decreasing function of the realized absolute utility difference between the two options. If one interprets random utility as reflecting fluctuating tastes or unobserved heterogeneity, this modeling reflects that the utility difference and hence the difficulty of the choice problem varies over time or between individuals. Choice and response time are then determined by the realization of the utility difference at a particular point in time or for a randomly chosen individual. Alternatively, one can also interpret random utility as a consequence of perceptual noise. In this case, our modeling can be given a foundation from the perspective of the evidence accumulation models used in psychology and neuroscience \citep{Ratcliffetal16,ShadlenKiani13} and recently economics \citep{FudenbergStrackStrzalecki18}. We will discuss the connection between our approach and these models below.

To provide an intuition for our results, consider the choice between two options $x$ and $y$, where $x$ is chosen with probability $p$ and $y$ with probability $1-p$. For the sake of clarity, let us adopt the interpretation that these probabilities describe the choices of a single individual across many repetitions of the problem. We will first show that observing $p > 1/2$ is not sufficient to conclude that the individual prefers $x$ to $y$, i.e., that the underlying utility of $x$ is larger than that of $y$, if no assumptions about the shape of the utility distribution are made. Specifically, it is possible to rationalize the data by a random utility model (RUM) that has a deterministic utility function with $u(x) < u(y)$ and asymmetric noise with zero mean. The asymmetry is such that the realized utility difference between $x$ and $y$ is often positive, generating $p > 1/2$, but takes large absolute values whenever it is negative. We will show that asymmetric distributions in fact arise very naturally, for instance in random parameter models that have recently become prominent \citep{ApesteguiaBallesterLu17,ApesteguiaBallester18}. Wrongfully assuming symmetry then leads to false inferences about preferences.

Now assume we have data on the joint distribution of choices and response times. We will show in Theorem \ref{Thm.RUMRT} that $p > 1/2$ combined with, informally speaking, a comparatively slow choice of $y$ relative to $x$ is sufficient to conclude that the individual prefers $x$ to $y$, even without making any assumptions about the shape of the utility distribution. A slow choice of $y$ relative to $x$ reveals that the utility difference cannot be distributed too asymmetrically in the way described above, because negative utility differences with large absolute values would generate quick choices of $y$, based on the chronometric relationship. Importantly, this argument does not presume knowledge of the shape of the chronometric function beyond monotonicity (and some technical properties). More formally, let $F(x)(t)$ and $F(y)(t)$ be the cumulative distribution functions of response times conditional on the choice of $x$ and $y$, respectively. Our criterion states that when we observe
\[ F(y)(t) \leq \frac{p}{1-p} F(x)(t) \text{ for all } t \geq 0,\]
then any random utility model with a chronometric function (RUM-CF) which rationalizes the data must satisfy $u(x) \geq u(y)$. A similar statement holds for strict preferences. In the limit as $p \rightarrow 1/2$, that is, as choice data alone becomes uninformative, this becomes the condition that choice of $y$ must be slower than choice of $x$ in the first-order stochastic dominance sense.  As $p$ grows, hence choice data becomes more indicative of a preference, the condition becomes weaker and only requires that choice of $y$ is not much faster than choice of $x$. 

Our result provides a revealed preference criterion for the analyst who is reluctant to make distributional assumptions in random utility models. Based on observed response times, it is often possible to deduce preferences without such assumptions. This avoids making mistakes like those that arise when symmetry is wrongfully imposed. Despite its cautiousness, the criterion is tight and fully recovers preferences from a large class of data sets, as we will argue below. We will also argue that the criterion can sometimes arbitrate in cases of stochastically inconsistent choices.

We then study the case where the analyst has reasons to believe that utility differences are symmetrically distributed, as is often assumed in the literature (e.g.\ in any application with logit or probit choice). It then follows immediately that $p > 1/2$ implies $u(x) > u(y)$, so preferences are revealed by choices without response times. But now we show that the use of response time data enables the identification of preferences for choice pairs outside the set of available choice data. For the case of deterministic choices and deterministic response times, this has been noted before. \citet{KrajbichOudFehr14} argue that a slow choice of $z$ over $x$ combined with a quick choice of the same $z$ over $y$ reveals a preference for $x$ over $y$, even though the choice between $x$ and $y$ is not directly observed and a transitivity argument is not applicable. The idea is that, based on the chronometric relationship, the positive utility difference $u(z)-u(x)$ must be smaller than the positive utility difference $u(z)-u(y)$, which implies $u(x)>u(y)$. To date, however, it has remained an open question how to implement this idea, since real-world choices and response times are stochastic, and hence it is unclear what ``choice of $z$ over $x$'' and ``slow versus fast'' exactly means. For instance, is ``faster than'' defined in terms of mean response times, median response times, or some other characteristic of the response time distribution? Our Theorem \ref{Thm.SRUMRT} provides an answer to that question. Suppose $z$ is chosen over $x$ with a probability $p(z,x)>1/2$, which indeed implies $u(z)-u(x)>0$ in the symmetric-noise case. Then we define $t(z,x)$ as a specific percentile of the response time distribution for $z$, namely the $(1-p(z,x))/p(z,x)$-percentile. Analogously, if $z$ is chosen over $y$ with a probability $p(z,y)>1/2$, which implies  $u(z)-u(y)>0$, the corresponding percentile $t(z,y)$ can be defined. Our result shows that these observable percentiles are the appropriate measure of preference intensity for the stochastic setting, in the sense that $t(z,x) > t(z,y)$ implies $u(z)-u(x) < u(z)-u(y)$ and hence a revealed (strict) preference for $x$ over $y$. That is, inference cannot be based on mean, median, maximum, or minimum response times. The correct measurement is the $(1-p)/p$-percentile of the distribution of response times, which in particular requires using a different percentile from each choice pair, adjusting for the respective choice frequencies. This is a quantification which comes out of the analytical model and which, to the best of our knowledge, has never been empirically utilized. Yet, since a revealed strict preference for $x$ over $y$ translates into a choice probability $p(x,y)>1/2$ when the utility distribution is symmetric, it generates out-of-sample predictions that should be easy to test empirically.

In the traditional approach without response times, making out-of-sample predictions requires even stronger distributional assumptions than just symmetry. Random utility models like probit or logit are instances of Fechnerian models \citep{Debreu58,Moffatt15}, in which the utility difference between the two options follows the exact same distributional form in all binary choice problems. With this Fechnerian assumption (but without assuming a specific distributional form), already the choice observation $p(z,x) < p(z,y)$ reveals a preference for $x$ over $y$. Put differently, the Fechnerian assumption enables an exhaustive elicitation of ordinal preferences even outside the data set. But now we show that the use of response time data makes it possible to move beyond ordinal preferences and make predictions of precise choice probabilities. Theorem \ref{Thm.FRUMRT} provides a closed-form formula to predict $p(x,y)$ based on observables, i.e., choice probabilities and response times, from only the binary choices between $z$ and $x$ and between $z$ and $y$. 

The general pattern that emerges from our results is that response time data allow us to obtain results that would otherwise require an additional distributional assumption that might be empirically unjustified. Response time data make it possible to get rid of assumptions because the distribution of response times contains information about the distribution of utility. This enables the revelation of preferences without any distributional assumptions, makes it possible to extrapolate preferences to cases for which no choice data exist with a symmetry assumption, and even generates precise probability predictions with the Fechnerian assumption.

Our Theorem \ref{Thm.RUMRT} provides a robust sufficient condition for preference revelation, which essentially goes ``from data to models.'' To investigate how much bite our criterion has, we also look at the converse implication ``from models to data,'' i.e., we study stochastic choice functions with response times (SCF-RTs) that are generated by standard models from the received literature. Put differently, we now take a specific data-generating process as given and apply our agnostic method that does not presume knowledge of the process to the resulting data set. We do this first for the whole class of RUM-CFs that have symmetric distributions, which contains the (generalized) probit and logit models as special cases but goes far beyond them. We show that our criterion recovers all preferences correctly when any such model actually generated the data. In other words, our sufficient condition is also necessary and has maximal bite for the entire class of SCF-RTs generated by symmetric RUM-CFs. Even the analyst who believes in the probit or logit distribution can work with our criterion, because it must always hold in his data. Our approach will yield the same revealed preferences as an application of the full-fledged model if his belief is correct, but avoids a mistake if his belief is incorrect. We show that this is still the case with additional noise in response times, as long as the noise is from an independent source (like a stochastic chronometric function or imperfect observation) and does not systematically reverse the chronometric relationship.

Second, we study the class of drift-diffusion models (DDMs) with constant or collapsing decision boundaries, which are prominent in psychology and neuroscience \citep[e.g.][]{Ratcliff78,ShadlenKiani13}. These models have recently attracted attention in economics because they can be derived from optimal evidence accumulation mechanisms \citep{Drugowitschetal12,TajimaDrugowitschPouget16,FudenbergStrackStrzalecki18}. We show that our criterion again recovers all preferences correctly from data that is generated by a DDM. This is almost immediate for the classical case of constant boundaries, but it is also a property of the case of collapsing boundaries. Hence our previous statement on believers in probit or logit models also applies to believers in drift-diffusion models. Furthermore, in the case of a DDM with collapsing boundaries, there is a surprising and so far unnoticed connection between the decision boundary and the chronometric function: one can be interpreted as the inverse of the other. 

The paper is structured as follows. Section \ref{Sec.Model} presents the formal setting. Section \ref{Sec.Reveal} develops the main results, devoting separate subsections to the unrestricted, symmetric, and Fechnerian cases. Section \ref{Sec.Specific} shows that choice data generated by standard models from economics and psychology fulfills our main criterion for preference revelation. Section \ref{Sec.Lit} discusses the related literature in more detail, and Section \ref{Sec.Conclusion} concludes. All proofs omitted from the main text can be found in the appendix.

\section{Formal Setting and Definitions}\label{Sec.Model}

Let $X$ be a finite set of options. Denote by $C=\set{(x,y)}{x,y\in X, x\ne y}$ the set of all binary choice problems, so $(x,y)$ and $(y,x)$ both represent the problem of choice between $x$ and $y$. Let $D\subseteq C$ be the set of choice problems on which we have data, assumed to be non-empty and symmetric, that is, $(x,y)\in D$ implies $(y,x)\in D$. To economize notation, we let the set $D$ be fixed throughout.

\begin{definition}\label{defscf}
A {\em stochastic choice function} (SCF) is a function $p$ assigning to each $(x,y) \in D$ a probability $p(x,y)>0$, with the property that $p(x,y)+p(y,x)=1$.
\end{definition}

In an SCF, $p(x,y)$ is interpreted as the probability of choosing $x$ in the binary choice between $x$ and $y$, and $p(y,x)$ is the probability of choosing $y$. The assumption that $p(x,y)>0$ for all $(x,y) \in D$ implies that choice is stochastic in a non-degenerate sense, because each alternative is chosen with strictly positive probability.

Since there is no universally agreed-upon definition of random utility models, we will work with a fairly general definition which encompasses many previous ones. In particular, it is convenient for our analysis to directly describe for each $(x,y) \in C$ the distribution of the utility difference between the two options.

\begin{definition}\label{def.rum}
A {\em random utility model} (RUM) is a pair $(u,g)$ where $u: X \rightarrow \RR$ is a utility function and $g$ assigns a density function $g(x,y)$ on $\RR$ to each $(x,y) \in C$, fulfilling the following properties:

\begin{enumerate}[(RUM.3)]

\item[(RUM.1)] $\int_{-\infty}^{+\infty} v g(x,y)(v) d v = u(x) - u(y)$,

\item[(RUM.2)] $g(x,y)(v)=g(y,x)(-v)$ for all $v \in \RR$, and

\item[(RUM.3)] the support of $g(x,y)$ is connected.

\end{enumerate}
\end{definition}

In a RUM, the utility function $u$ represents the underlying preferences which the analyst aims to uncover, while the density $g(x,y)$ also incorporates the noise, that is, it describes the distribution of the random utility difference between $x$ and $y$. We denote the corresponding cumulative distribution function by $G(x,y)$. The noise in each option's utility is assumed to have zero mean, so the expected value of the random utility difference between $x$ and $y$ must be $u(x)-u(y)$, as required by (RUM.1). We will also use the notation $v(x,y)=u(x)-u(y)$. Condition (RUM.2) states that $g(x,y)$ and $g(y,x)$ describe the same random utility difference but with opposite sign. (RUM.3) is a regularity condition stating that there are no gaps in the distribution of an option's utility.

Our definition reflects the conventional idea that RUMs consist of a deterministic utility function plus stochastic error terms, as typically implemented in microeconometrics \citep[see][for a history of the approach]{McFadden01}. It is more general than conventional models because the densities $g(x,y)$ are unrestricted across choice pairs. This allows us to accommodate pair-specific factors other than utility differences that may affect choice probabilities (and response times), such as obvious dominance relations between some options \citep[see e.g.][and our discussion in Section \ref{Subsec.LitSC}]{HeNatenzon18}. An alternative definition views RUMs as a distribution over deterministic utility functions \citep[see, e.g.][]{GulPesendorfer06,GulNatenzonPesendorfer14}.\footnote{This has its roots in the work of \citet{BlockMarschak60}, who posed the question of when a stochastic choice function can be rationalized by a probability distribution over preference orderings, the idea being that one ordering is realized before every choice.} This approach is also less general than ours, because not every collection $g$ of utility difference distributions can be generated by a distribution over the set of deterministic utility functions \citep{Falmagne78,BarberaPattanaik86}. An additional benefit of our generality is that our positive results on preference revelation become stronger, because they hold within a larger class of models. The added generality does not matter for the impossibility result in Proposition \ref{Prop.RUM}, which would still hold for RUMs defined as distributions over deterministic utility functions.

A RUM generates choices by assuming that the option chosen is the one with the larger realized utility. Specifically, given a RUM $(u,g)$ and a pair $(x,y)\in C$, the probability that the utility of $y$ exceeds that of $x$ is $G(x,y)(0)$. This motivates the following definition.

\begin{definition}\label{defratrum}
A RUM $(u,g)$ {\em rationalizes} an SCF $p$ if $G(x,y)(0)=p(y,x)$ holds for all $(x,y) \in D$.
\end{definition}

We now extend the framework and include response times, by adding conditional response time distributions for each choice. This is the easiest way of describing a joint distribution over choices and response times.

\begin{definition}
A {\em stochastic choice function with response times} (SCF-RT) is a pair $(p,f)$ where $p$ is an SCF and $f$ assigns to each $(x,y) \in D$ a strictly positive density function $f(x,y)$ on $\RR^+$. 
\end{definition}

The density $f(x,y)$ describes the distribution of response times conditional on $x$ being chosen in the binary choice between $x$ and $y$. The corresponding cumulative distribution function is denoted by $F(x,y)$. It would be straightforward to introduce lower or upper bounds on response times, for instance due to a non-decision time or a maximal observed response time. We refrain from doing so here for notational convenience and comparability with the literature \citep[e.g.][]{FudenbergStrackStrzalecki18}.

\begin{definition}\label{defrumrt}
A {\em random utility model with a chronometric function} (RUM-CF) is a triple $(u,g,r)$ where $(u,g)$ is a RUM and $r: \RR^{++} \rightarrow \RR^+$ is a continuous function that is strictly decreasing in $v$ whenever $r(v)>0$, with $\lim_{v \rightarrow 0} r(v) = \infty$ and $\lim_{v \rightarrow \infty} r(v) =0$.
\end{definition}

In a RUM-CF, $r$ represents the chronometric function. It maps utility differences $v$ into response times $r(|v|)$, such that larger absolute utility differences generate shorter response times. The assumption that $\lim_{v \rightarrow 0} r(v) = \infty$ and $\lim_{v \rightarrow \infty} r(v) =0$ ensures that the model can encompass all response times observed in an SCF-RT. Our definition allows for functions like $r(v)=1/v$ that are strictly decreasing throughout, and also for functions that reach $r(v)=0$ for large enough $v$. The latter case will arise when we construct chronometric functions from sequential sampling models in Section \ref{Sec.Specific}. Figure \ref{Fig.CF} illustrates both cases, taking advantage of the fact that the inverse $r^{-1}(t)$ is well-defined for the restriction of $r$ to the subset where $r(v)>0$.

\begin{figure}[!t]
  \centering
\psset{xunit=2cm,yunit=2.5cm,linewidth=2pt}
\begin{pspicture}(0,0)(6,1.5)
\psline[linewidth=1pt](0,0)(6,0)
\psline[linewidth=1pt](0,0)(0,1.5)
\psbezier(0,1)(4,0.75)(2,0.25)(6,0.1)

\psline[linestyle=dashed](2.75,0)(2.75,0.6)
\psline[linestyle=dashed](0,0.6)(2.75,0.6)

\rput[c](2.75,-0.15){$t$}
\rput[c](-0.25,0.6){$v$}
\rput[c](4.5,0.5){$r^{-1}(\cdot)$}

\end{pspicture}  

\vspace{0.5cm}

\begin{pspicture}(0,0)(6,1.5)
\psline[linewidth=1pt](0,0)(6,0)
\psline[linewidth=1pt](0,0)(0,1.5)
\psbezier(0.1,1.5)(1,0)(2,0.2)(6,0.1)

\psline[linestyle=dashed](1,0)(1,0.56)
\psline[linestyle=dashed](0,0.56)(1,0.56)

\rput[c](1,-0.1){$t$}
\rput[c](-0.25,0.56){$v$}
\rput[c](4,0.4){$r^{-1}(\cdot)$}

\end{pspicture} 
\caption{Two illustrations of chronometric functions $r$, mapping realized utility differences (vertical axis) into response times (horizontal axis).}\label{Fig.CF}
\end{figure}
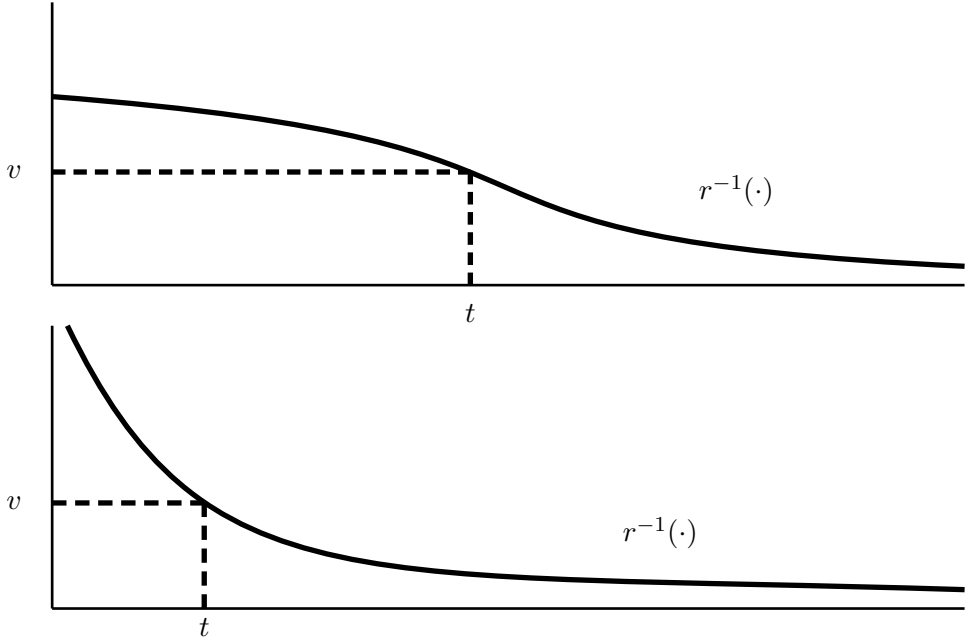

In addition to choices, a RUM-CF generates response times by assuming that the realized response time is related to the realized utility difference through the function $r$. Specifically, given a RUM-CF $(u,g,r)$ and a pair $(x,y)\in C$, the probability of a response time of at most $t>0$, conditional on $x$ being chosen over $y$, is the probability that the realized utility difference is at least $r^{-1}(t)$, conditional on that difference being positive. This probability can be calculated as $\left[1-G(x,y)(r^{-1}(t))\right]/\left[1-G(x,y)(0)\right]$, which motivates the following definition.

\begin{definition}\label{defratscfrt}
A RUM-CF $(u,g,r)$ {\em rationalizes} an SCF-RT $(p,f)$ if $(u,g)$ rationalizes $p$ and
\begin{align}\label{ruratsc}\frac{1-G(x,y)\left(r^{-1}(t)\right)}{1-G(x,y)(0)}= F(x,y)(t)\end{align}
holds for all $t > 0$ and all $(x,y) \in D$.
\end{definition}

An alternative approach would have been to assume that response time is a decreasing function of the underlying absolute utility difference $|v(x,y)|$ between the two options (as opposed to the realized, noisy ones). A drawback of this approach is that response times would be predicted to be deterministic, in contradiction to all available evidence. Hence a second source of noise would have to be introduced, for instance by making the chronometric function stochastic. Without additional, {\em ad hoc} assumptions, any such model would predict that the conditional distributions of response times for each of the two choices are identical, a further prediction not borne out by the data. Our approach is more parsimonious as it requires only one source of randomness (in utility) to generate both stochastic choices and stochastic response times, and it does not make the implausible prediction of independence between choices and response times. That said, a second source of independent noise (like a stochastic chronometric function) can be introduced without changing our main insights, as we will show in Section \ref{Sec.Specific}.

When studying stochastic choice data, the analyst might be interested in or willing to make specific assumptions about the distribution of random utility, i.e., on the properties of the collection of densities $g$. In doing so, the analyst accepts a restriction to a specific subclass of random utility models. Those might range from symmetry of each $g(x,y)$ to specific functional forms. We say that an SCF, respectively an SCF-RT, is {\em rationalizable} within some class of models if there exists a RUM, respectively a RUM-CF, in that class that rationalizes it.

\begin{definition}
Within a class of models, a rationalizable SCF (SCF-RT) \textit{reveals a preference for $x$ over $y$} if all RUMs (RUM-CFs) in the class that rationalize it satisfy $u(x)\geq u(y)$. It \textit{reveals a strict preference for $x$ over $y$} if all RUMs (RUM-CFs) in the class that rationalize it satisfy $u(x) > u(y)$.
\end{definition}

\section{Revealed Preference}\label{Sec.Reveal}

In this section, we investigate the use of response times for preference revelation. Specifically, we are interested in preference revelation within difference classes of random utility models, and how the addition of response times improves the results.

\subsection{The Unrestricted Case}\label{Subsec.Unrestricted}

The first observation is that, without further restrictions on the utility distributions, and without the use of response times, nothing can be learned from choice probabilities. This is well-known among specialists in stochastic choice theory and hence we do not claim originality.

\begin{proposition}\label{Prop.RUM}
Within the class of all RUMs, a rationalizable SCF reveals no preference between any $x$ and $y$ with $x \ne y$.
\end{proposition}

The intuition for the result is simple. Data on the choice between $x$ and $y$ allows us to learn the value $G(x,y)(0)$, but, without distributional assumptions, this does not tell us whether the expected value $v(x,y)=u(x)-u(y)$ is positive or negative, which is what we are interested in.

A solution to the problem would be to impose the seemingly innocuous assumption of symmetry of the distribution. In that case, $G(x,y)(0) \leq 1/2$ indeed implies $v(x,y)\geq 0$, and $G(x,y)(0) \geq 1/2$ implies $v(x,y) \leq 0$. We will investigate the symmetry assumption, and the scope for response times to improve preference revelation under that assumption, in Section \ref{Subsec.Symmetric}. First, however, we want to illustrate with a simple example that symmetry may sometimes not be an innocuous assumption at all.

\vspace{0.3cm}

\begin{example}\label{Ex.1}
Consider a special case where our abstract set of options $X$ consists of lotteries over monetary outcomes. In particular, consider a choice between two options $x$ and $y$, where $x$ is a lottery that pays $20$ with probability $1/20$ and $0$ with opposite probability $19/20$, and $y$ is a safe option that pays $1$ with probability $1$. The data-generating process is a RUM based on CRRA expected utility with a stochastic risk-aversion parameter $\alpha$ \citep[see][for a recent treatment of random parameter models]{ApesteguiaBallester18}. Given a realized value of $\alpha$, the utility of $x$ is $u(x|\alpha)=(1/20)20^{\alpha}$ and the utility of $y$ is $u(y|\alpha)=1$, so the realized utility difference is $v(x,y|\alpha)=u(x|\alpha)-u(y|\alpha)=20^{\alpha-1} -1$. Suppose the risk-aversion parameter follows a uniform distribution $\alpha \sim U[0.4,1.4]$. Then $x$ will be chosen if $\alpha>1$ and $y$ will be chosen if $\alpha<1$, which yields the SCF $p(x,y)=0.4$ and $p(y,x)=0.6$. The cdf of the entire utility difference distribution can be calculated as
\[G(x,y)(v)=\text{Prob}[v(x,y|\alpha) \leq v]=0.6 + \frac{\ln(v+1)}{\ln 20},\]
for all $v$ in the support $[\underline{v},\overline{v}]$ given by $\underline{v} \approx -0.83$ and $\overline{v} \approx 2.31$. The pdf is
\[g(x,y)(v)=\frac{1}{(v+1)\ln 20},\] which is clearly not symmetric but strictly decreasing in $v$.

If the analyst knew the data-generating process, he would accept this asymmetry as a natural characteristic of the specific environment. Consider, however, an analyst who does not know the data-generating process. Suppose that this analyst erroneously imposes the symmetry assumption when analyzing the SCF with the intention of uncovering the underlying utility. This analyst will correctly deduce $G(x,y)(0)=0.6$, but then, applying symmetry, incorrectly conclude that $v(x,y)$ is strictly negative, i.e., that the SCF reveals a strict preference for $y$ over $x$. The true data-generating process, however, satisfies \[v(x,y)=\int_{0.4}^{1.4} v(x,y|\alpha) d\alpha \approx 0.05 > 0,\] i.e., the average utility of $x$ is strictly larger than the average utility of $y$.

This example is meant to illustrate $(i)$ that symmetry is not always a plausible restriction, in particular when noise is non-linear as in many random parameter models, and $(ii)$ that erroneously making the symmetry assumption can lead to wrong inferences about preferences.
\end{example}

\vspace{0.3cm}

The following result shows that, if response times are available, it may become possible to learn a preference even in the unrestricted class of models. We first introduce the following new concepts. Given two cumulative distribution functions $G$ and $H$ on $\RR^+$ and a constant $q \geq 1$, we say that \textit{$G$ $q$-first-order stochastically dominates $H$} (also written $G$ $q$-FSD $H$) if \[G(t) \leq q \cdot H(t) \text{ for all } t\geq 0.\] If, additionally, the inequality is strict for some $t$, then \textit{$G$ strictly $q$-first-order stochastically dominates $H$} (written $G$ $q$-SFSD $H$). For $q=1$, these concepts coincide with the standard notions of first-order stochastic dominance. They are weaker requirements when $q>1$, and possibly substantially so, because the dominating function $G$ can lie above $H$ to an extent only constrained by the ratio $q$. In particular, $q$-FSD implies $q'$-FSD whenever $q\leq q'$. Furthermore, for any two distributions $G$ and $H$ for which $G(t)/H(t)$ is bounded, we can always find a large enough $q$ such that $G$ $q$-FSD $H$.

\begin{theorem}\label{Thm.RUMRT}
Within the class of all RUM-CFs, a rationalizable SCF-RT reveals a preference for $x$ over $y$ if $F(y,x)$ $q$-FSD $F(x,y)$, and a strict preference if $F(y,x)$ $q$-SFSD $F(x,y)$, for $q=p(x,y)/p(y,x)$.
\end{theorem}

\begin{proof} Let $(u,g,r)$ be any RUM-CF which rationalizes an SCF-RT $(p,f)$, and consider any $(x,y) \in D$. By (\ref{ruratsc}), it holds that \begin{align}\label{eq.pa}1 - G(x,y)(r^{-1}(t))=p(x,y) F(x,y)(t) \end{align}
for all $t > 0$. Since (RUM.2) implies $1 - G(y,x)(v)=G(x,y)(-v)$, for $(y,x) \in D$ we analogously obtain
\begin{align}\label{eq.pb}G(x,y)(-r^{-1}(t))= p(y,x) F(y,x)(t)\end{align}
for all $t > 0$. With the definition of $Q(x,y)(t)=(p(x,y) F(x,y)(t))/(p(y,x) F(y,x)(t))$ we therefore have
\begin{align}\label{eq.p1}
Q(x,y)(t)=\frac{1 - G(x,y)(r^{-1}(t))}{G(x,y)(-r^{-1}(t))}
\end{align}
for all $t>0$.

Now suppose $F(y,x)$ $q$-FSD $F(x,y)$ for $q=p(x,y)/p(y,x)$. This can equivalently be written as $Q(x,y)(t) \geq 1$ for all $t>0$. Hence it follows from (\ref{eq.p1}) that
\[G(x,y)(-r^{-1}(t)) \leq 1- G(x,y)(r^{-1}(t))\]
for all $t>0$. We claim that this implies
\begin{align}\label{eq.p2}G(x,y)(-v) \leq 1- G(x,y)(v)\end{align}
for all $v\geq0$. The inequality follows immediately for any $v$ for which there exists $t>0$ such that $r^{-1}(t)=v$. For $v=0$ it follows from continuity of $G(x,y)(v)$. For any $v$ with $r(v)=0$ it follows because in that case $G(x,y)(v)=1$ and $G(x,y)(-v)=0$, as otherwise the RUM-CF would generate an atom at the response time of zero.

Define a function $H: \mathbb{R} \rightarrow [0,1]$ by \[H(v)= \begin{cases} 1 - G(x,y)(-v) & \text{if } v \geq 0,\\ G(x,y)(v) & \text{if } v < 0. \end{cases}\] Observe that $H$ is the cumulative distribution function for a distribution that is symmetric around zero and continuous except (possibly) for an atom at zero. Hence we have
\begin{align*}\int_{-\infty}^{+\infty} v d H(v) & = \int_{(-\infty,0)} v d H(v) + \int_{(0,+\infty)} v d H(v) \\ & = \int_{-\infty}^{0} v g(x,y)(v) dv + \int_{0}^{+\infty} v g(x,y)(-v) dv \\ & = \int_{-\infty}^{0} v g(x,y)(v) dv - \int_{-\infty}^{0} v g(x,y)(v) dv =0. \end{align*}
Observe furthermore that (\ref{eq.p2}) implies $G(x,y)$ $1$-FSD $H$. Hence we have
\begin{align}\label{eq.p3} v(x,y)=\int_{-\infty}^{+\infty} v dG(x,y)(v) \geq \int_{-\infty}^{+\infty} v dH(v)=0,\end{align} i.e., a revealed preference for $x$ over $y$.

If $F(y,x)$ $q$-SFSD $F(x,y)$ for $q=p(x,y)/p(y,x)$, then (\ref{eq.p2}) is strict for some $v\geq0$. Hence $G(x,y)$ $1$-SFSD $H$ and the inequality in (\ref{eq.p3}) is strict, i.e., a revealed strict preference for $x$ over $y$.\end{proof}

The basic idea behind Theorem \ref{Thm.RUMRT} is that the observable distributions of response times provide information about the unobservable distributions of utilities, based on the chronometric relationship.

To understand the precise condition, assume first that $q=p(x,y)/p(y,x)=1$ for some $(x,y) \in D$, i.e., both options are equally likely to be chosen. Any RUM-CF that rationalizes this choice must satisfy $G(x,y)(0)=1/2$. Furthermore, note that the distribution of $v>0$ generates $F(x,y)$ and the distribution of $v<0$ generates $F(y,x)$. Thus, if we additionally observe that $F(x,y)(t) = F(y,x)(t)$ for all $t \geq 0$, i.e., identical response time distributions for the two options, then we can conclude that the shape of the utility difference distribution must be identical on the positive and on the negative domain. This requires no knowledge of the properties of $r$ beyond monotonicity. Hence we have \textit{verified} that the distribution is symmetric around zero, so its mean $v(x,y)$ is zero. Our theorem indeed implies a revealed preference for $x$ over $y$ and for $y$ over $x$ in this case, which we also call a \textit{revealed indifference between $x$ and $y$}. If, by contrast, we observe that $F(y,x)$ $1$-SFSD $F(x,y)$, i.e., the choice of $y$ is systematically slower than the choice of $x$, we can conclude that the utility difference distribution is asymmetric and takes systematically larger absolute values on the positive than on the negative domain. Hence its mean $v(x,y)$ is strictly larger than zero, which translates into a revealed strict preference for $x$ over $y$. Finally, if we observe $q = p(x,y)/p(y,x)>1$, then to obtain a revealed preference for $x$ over $y$ it is sufficient that the choice of $y$ is not too much faster than the choice of $x$, as captured by our concept of $q$-first-order stochastic dominance. If choice behavior is already indicative of a particular preference, then the response time distributions just need to confirm that the utility difference distribution is not strongly asymmetric in the reverse direction.\footnote{Theorem \ref{Thm.RUMRT} and our further results focus on uncovering the (sign of the) mean of $G(x,y)$, because the mean equals $u(x)-u(y)$ and therefore informs about the (ordinal) preferences represented by $u$, for either normative or positive reasons. Potentially, one could be interested in uncovering also other summary statistics of $G(x,y)$, and our tools may be helpful for that purpose, but the relevance of other statistics is not obvious from the point of view of revealed preference theory.}

In the remainder of the section, we will discuss several implications of Theorem \ref{Thm.RUMRT}. First, it is a frequent empirical observation that, in many decision situations where objectively correct and wrong responses exist, errors are slower than correct responses.\footnote{See \citet{Luce86} for a discussion of the classical evidence. Of course, the picture is complicated if decisions are subject to extraneous impulsive tendencies, as e.g.\ alternative decision processes reflecting underlying biases. For the implications of the latter for response times, see \citet{BayesRT}. See also Section \ref{Subsec.LitSC} for a discussion.} Our result shows that a similar logic holds for preferential choice. While the definition of error and correct response is not obvious ex-ante for preferential choice, ex-post we can call the choice of $y$ a \textit{revealed error} and the choice of $x$ a \textit{revealed correct response} when $x$ is revealed to be strictly preferred over $y$. Translated into this language, it follows immediately that slow choice in the standard first-order stochastic dominance sense indeed always reveals an error, provided that choice probabilities are at least minimally informative.

Second, we have described SCF-RTs by unconditional choice probabilities and conditional response time distributions for each choice. This is the natural extension of SCFs and allowed us to work out the intuition for our result. Alternatively, we could have described the joint distribution over choices and response times by an unconditional response time distribution and conditional choice probabilities for each response time. Let $P(x,y)(t)$ denote the probability of a choice of $x$ over $y$ conditional on choice taking place before time $t$. The ratio of these probabilities can be calculated as \[Q(x,y)(t)=\frac{P(x,y)(t)}{P(y,x)(t)}=\frac{p(x,y)F(x,y)(t)}{p(y,x)F(y,x)(t)}.\] Hence, the condition that $F(y,x)$ $q$-FSD $F(x,y)$ for $q=p(x,y)/p(y,x)$ in Theorem \ref{Thm.RUMRT} is equivalent to \begin{align}\label{Eq.Q}Q(x,y)(t) \geq 1 \text{ for all } t > 0.\end{align} An analogous formulation holds for the strict case. For a revealed preference without distributional assumptions, we can thus also check if $x$ is more likely to be chosen than $y$ \textit{before all times $t$}. The simple requirement $p(x,y) \geq p(y,x)$ obtains as a special case of this in the limit as $t \rightarrow \infty$.

The formulation based on $Q(x,y)(t)$ suggests a natural stronger condition. Let $p(x,y)(t)$ denote the probability of a choice of $x$ over $y$ conditional on choice taking place at time $t$ (rather than before $t$). We obtain the ratio \[q(x,y)(t)=\frac{p(x,y)(t)}{p(y,x)(t)}=\frac{p(x,y)f(x,y)(t)}{p(y,x)f(y,x)(t)},\] and can state the following corollary to Theorem \ref{Thm.RUMRT}.

\begin{corollary}\label{Cor.RUMRT1}
Within the class of all RUM-CFs, a rationalizable SCF-RT reveals a preference for $x$ over $y$ if $q(x,y)(t) \geq 1$ for almost all $t \geq 0$, and a strict preference if, additionally, the inequality is strict for a set of $t$ with positive Lebesgue measure.
\end{corollary}

The condition that $x$ is more likely to be chosen than $y$ \textit{at almost all times $t$} can be interpreted as a requirement of stochastic consistency across all response times. This is clearly stronger than (\ref{Eq.Q}). Appendix \ref{App.Additional} contains an example in which $p(x,y)(t) < p(y,x)(t)$ holds for an interval of response times, so Corollary \ref{Cor.RUMRT1} is not applicable, but a strict preference for $x$ over $y$ is still revealed by Theorem \ref{Thm.RUMRT}. Hence our main criterion arrives at a conclusion even though behavior displays stochastic inconsistency across response times.\footnote{Notice a similarity to \citet{BernheimRangel09}, who require agreement of choices across choice sets or frames to obtain a revealed preference. A first difference is that we study stochastic choice and contemplate probabilistic agreement of choices across response times. A second difference is that our main criterion can reveal a preference even if there is no such agreement.}

Finally, the result in Theorem \ref{Thm.RUMRT} can be extended by completing revealed preferences in a transitive way. Collect in a binary relation $R^{rt}$ on $X$ all the preferences that are directly revealed, by defining
\[(x,y) \in R^{rt} \; \Leftrightarrow \; F(y,x) \; q\text{-FSD} \; F(x,y) \; \text{for} \; q=\frac{p(x,y)}{p(y,x)}, \; \text{or} \; x=y.\]
For any binary relation $R$ on $X$, denote by $T(R)$ the transitive closure of $R$, i.e., $(x,y) \in T(R)$ if and only if there exists a sequence $x_1,x_2,\ldots,x_n$ of any length $n \geq 2$ with $x_1=x$, $x_n=y$ and $(x_k,x_{k+1}) \in R$ for all $k=1,\ldots,n-1$. Let $P$ denote the asymmetric part of $R$, and $T_P(R)$ the asymmetric part of $T(R)$. With this notation, we obtain the following result.

\begin{corollary}\label{Cor.RUMRT2}
Within the class of all RUM-CFs, a rationalizable SCF-RT reveals a preference for $x$ over $y$ if $(x,y) \in T(R^{rt})$, and a strict preference if $(x,y) \in T_P(R^{rt})$.
\end{corollary}

The results in this section are interesting for two main reasons. First, for the analyst who is reluctant to make distributional assumptions in the context of random utility models, Theorem \ref{Thm.RUMRT} provides a robust criterion for preference revelation. The criterion may lead to an incomplete revelation of preferences (we will return to this issue in Section \ref{Sec.Specific}), but it avoids making mistakes like those illustrated in Example \ref{Ex.1}. Second, our criterion may be able to arbitrate if choice behavior violates stochastic transitivity \citep{Tversky69,RieskampBusemeyerMellers06,Tsetsosetal16}. For example, assume we observe a stochastic choice cycle with $p(x,y)>1/2$, $p(y,z)>1/2$, and $p(z,x)>1/2$. Such a cycle (and the associated response times) cannot be rationalized by any model with symmetric utility distributions, but it may be rationalizable by a model with asymmetric utility distributions. In that case, at most two of the three binary choices can reveal a preference, so that the remaining binary choice would be revealed to be misleading by transitivity. A similar argument applies if choices are affected by framing and we observe $p^f(x,y)>p^f(y,x)$ under frame $f$ but $p^{f'}(x,y)<p^{f'}(y,x)$ under frame $f'$. Again, our response time criterion may be able to detect which frame induces choices that are probabilistically more in line with the true preferences.

\subsection{The Symmetric Case}\label{Subsec.Symmetric}

The assumption of symmetry is often accepted in the literature. Formally, a RUM $(u,g)$ or RUM-CF $(u,g,r)$ is {\em symmetric} if each density $g(x,y)$ is symmetric around its mean $v(x,y)$, that is, if for each $(x,y) \in C$ and all $\delta \geq 0$, \[g(x,y)(v(x,y) + \delta)=g(x,y)(v(x,y) - \delta).\] In contrast to Proposition \ref{Prop.RUM}, this assumption allows to learn preferences from observed choice probabilities.

\begin{proposition}\label{Prop.SRUM}
Within the class of symmetric RUMs, a rationalizable SCF reveals a preference for $x$ over $y$ if $p(x,y) \geq p(y,x),$ and a strict preference if $p(x,y) > p(y,x)$.
\end{proposition}

This result is both simple and well-known, and we include a proof in the appendix only for completeness.\footnote{Its first statement in the economics literature that we are aware of is \citet{Manski77}, but an earlier, closely related statement for general stochastic choice can be found already in \citet{Edwards54}.} The result can again be extended by completing revealed preferences in a transitive way. Define a binary relation $R^{s}$ on $X$ by
\[(x,y) \in R^{s} \; \Leftrightarrow \; p(x,y) \geq p(y,x), \; \text{or} \; x=y.\]

\begin{corollary}\label{Cor.SRUM}
Within the class of symmetric RUMs, a rationalizable SCF reveals a preference for $x$ over $y$ if $(x,y) \in T(R^{s})$, and a strict preference if $(x,y) \in T_P(R^{s})$.
\end{corollary}

Note that relation $R^s$ is always more complete than relation $R^{rt}$, that is, $x R^{rt} y$ implies $x R^s y$. This means that every preference which can be learned with the help of response times without distributional assumptions can also be learned without response times at the price of making the (possibly unwarranted) symmetry assumption.\footnote{This statement holds for weak but not necessarily for strict preferences. We can indeed have $x P^{rt} y$ but $y R^{s} x$, in case $p(x,y)=p(y,x)=1/2$ and $F(y,x)$ $1$-SFSD $F(x,y)$. Any symmetric RUM that rationalizes such an SCF must have $v(x,y)=0$. However, there is no symmetric RUM-CF that rationalizes the SCF-RT, due to the asymmetric response times.  All rationalizing RUM-CFs must be asymmetric and have $v(x,y)>0$.} But even if one is willing to make the symmetry assumption, the addition of response times again improves what can be learned about preferences, as the following result will show. It is based on triangulating a preference indirectly through comparisons with a third option. For each $(x,y)\in D$ with $p(x,y)>p(y,x)$ define $t(x,y)$ as the $1/q$-percentile of the response time distribution of $x$, for $q= p(x,y)/p(y,x)$ as before, i.e.,
\[F(x,y)(t(x,y))=\frac{p(y,x)}{p(x,y)}.\]
The percentile $t(x,y)>0$ combines information about choice probabilities and response times. It becomes smaller as $p(x,y)$ becomes larger or as the choice of $x$ becomes faster in the usual first-order stochastic dominance sense. Hence a small value of $t(x,y)$ is indicative of a strong preference for $x$ over $y$. Comparison of these percentiles can make it possible to learn preferences for unobserved pairs $(x,y) \in C \setminus D$ even when transitivity is void.

\begin{theorem}\label{Thm.SRUMRT}
Within the class of symmetric RUM-CFs, a rationalizable SCF-RT reveals a preference for $x$ over $y$, where $(x,y) \in C \setminus D$, if there exists $z\in X$ such that $t(x,z) \leq t(y,z)$ or $t(z,x) \geq t(z,y)$, and a strict preference if $t(x,z) < t(y,z)$ or $t(z,x) > t(z,y)$.
\end{theorem}

\begin{proof} Let $(u,g,r)$ be any symmetric RUM-CF which rationalizes an SCF-RT $(p,f)$. We first claim that, for any $(x,y)\in D$ with $p(x,y) > p(y,x)$, it holds that $t(x,y)=r(2v(x,y))$. To see the claim, note that by rationalization and symmetry,
\begin{align}\label{symcon1} p(y,x)=G(x,y)(0)=1-G(x,y)(2v(x,y)).\end{align}
From (\ref{ruratsc}) we obtain 
\[p(x,y)F(x,y)(t)=1-G(x,y)(r^{-1}(t))\]
for all $t > 0$. Evaluated at $t=r(2v(x,y))$, which is well-defined because $v(x,y)>0$ by Proposition \ref{Prop.SRUM}, this yields 
\begin{align}\label{symcon2} p(x,y)F(x,y)(r(2v(x,y)))=1-G(x,y)(2v(x,y)).\end{align}
Combining (\ref{symcon1}) and (\ref{symcon2}) yields
\[F(x,y)(r(2v(x,y)))=\frac{p(y,x)}{p(x,y)},\]
and, by definition of $t(x,y)$, it follows that $t(x,y)=r(2v(x,y))$, proving the claim.

Consider now any $(x,y)\in C \setminus D$ for which there exists $z\in X$ with $t(x,z) \leq t(y,z)$, and hence $0 < r(2v(x,z)) \leq r(2v(y,z))$ by the above claim. Since $r$ is strictly decreasing in $v$ whenever $r(v)>0$, it follows that $u(x)-u(z)=v(x,z)\geq v(y,z)=u(y)-u(z)$ and hence $v(x,y)=u(x)-u(y)\geq 0$, i.e., a revealed preference for $x$ over $y$. If $t(x,z) < t(y,z)$, all the inequalities must be strict, so the revealed preference is strict. The case where $t(z,x) \geq t(z,y)$ or $t(z,x) > t(z,y)$ is analogous.\end{proof}

It has been observed before that response times can be used to infer preferences for unobserved choices. \citet{KrajbichOudFehr14} argue that a slow choice of $z$ over $x$ combined with a quick choice of the same $z$ over $y$ reveals a preference for $x$ over $y$, even though the choice between $x$ and $y$ is not directly observed and a transitivity argument is not applicable. Based on the chronometric relationship, the positive utility difference $u(z)-u(x)$ must be smaller than the positive utility difference $u(z)-u(y)$, which implies $u(x)>u(y)$. It remained unclear how to generalize the idea to a stochastic framework. Our Theorem \ref{Thm.SRUMRT} answers this question. The condition $t(z,x) \geq t(z,y)$ is the appropriate formulation of a stochastic choice of $z$ over $x$ being slower than a stochastic choice of $z$ over $y$. Of course, an analogous argument applies to a quick choice of $x$ over $z$ combined with a slow choice of $y$ over $z$, as formalized by our alternative condition $t(x,z) \leq t(y,z)$. Importantly, we need to compare specific percentiles of the response time distributions that depend on choice probabilities, and not, for example, just mean or maximum response times. 

Not too surprisingly, also the preferences revealed by Theorem \ref{Thm.SRUMRT} can be completed in a transitive way. We will state this as part of a more general corollary. Define a binary relation $R^{srt}$ on $X$ by \[ (x,y) \in R^{srt} \;\; \Leftrightarrow \;\; (x,y) \in C \backslash D \text{ and } \exists z \in X \text{ with } t(x,z) \leq t(y,z) \text{ or } t(z,x) \geq t(z,y).\]

\begin{corollary}\label{Cor.SRUMRT}
Within the class of symmetric RUM-CFs, a rationalizable SCF-RT reveals a preference for $x$ over $y$ if $(x,y) \in T(R^s \cup R^{srt})$, and a strict preference if $(x,y) \in T_P(R^s \cup R^{srt})$.
\end{corollary}

The results in this section enable first out-of-sample predictions. Consider an unobserved choice problem $(x,y) \in C \setminus D$. If, based on Corollary \ref{Cor.SRUMRT}, the SCF-RT reveals a strict preference for $x$ over $y$ in the class of symmetric RUM-CFs, then we predict that $p(x,y) > p(y,x)$, because each symmetric model with $v(x,y)>0$ generates such choice probabilities. If the SCF-RT reveals an indifference between $x$ and $y$, then we can even predict the precise probabilities $p(x,y)=p(y,x)=1/2$. Such predictions are easy to test empirically. In the next section, we will show that the predictions can be sharpened under a stronger assumption on the distribution of utilities.

\subsection{The Fechnerian Case}

Microeconometric models of random utility assume even more structure. For instance, the prominent probit or logit models are special cases of {\em Fechnerian} models, which go back to the representation result by \citet{Debreu58}. A RUM $(u,g)$ or RUM-CF $(u,g,r)$ is Fechnerian if there exists a common density $g$ that is symmetric around zero and has full support, i.e., $g(\delta)=g(-\delta)>0$ for all $\delta\geq 0$, such that, for each $(x,y) \in C$ and all $v \in \RR$,
\[g(x,y)(v)=g(v-v(x,y)).\]
In words, the utility difference distribution has the same shape for each pair $(x,y) \in C$ and is just shifted so that its expected value becomes $v(x,y)$. This additional structure makes it possible to deduce preferences through comparison with a third option, relying only on choice probabilities.

\begin{proposition}\label{Prop.FRUM}
Within the class of Fechnerian RUMs, a rationalizable SCF reveals a preference for $x$ over $y$, where $(x,y) \in C \setminus D$,  if there exists $z\in X$ such that $p(x,z) \geq p(y,z)$, and a strict preference if $p(x,z) > p(y,z)$.
\end{proposition}

As in the case of Proposition \ref{Prop.SRUM}, this result is well-known and we provide a short proof in the appendix only for completeness.\footnote{The argument can be traced back to \citet{Fechner1860} and \citet{Thurstone27}. Within economics, it has been spelled out e.g.\ in \citet{BallingerWilcox97}.} The transitive closure extension can also be easily obtained. Define a binary relation $R^f$ on $X$ by \[(x,y) \in R^f \; \Leftrightarrow \, (x,y) \in C \setminus D \text{ and } \exists z \in X \text{ with } p(x,z) \geq p(y,z).\]
Then, combining Propositions \ref{Prop.SRUM} and \ref{Prop.FRUM} yields the following result.

\begin{corollary} \label{Cor.FRUM} Within the class of Fechnerian RUMs, a rationalizable SCF reveals a preference for $x$ over $y$ if $(x,y) \in T(R^s \cup R^f)$, and a strict preference if $(x,y) \in T_P(R^s \cup R^f)$.
\end{corollary}

Relation $R^f$ contains a statement about an unobserved pair $(x,y) \in C \setminus D$ whenever $(x,z),(y,z) \in D$ for some third option $z$. Hence, imposing the Fechnerian assumption enables an exhaustive elicitation of ordinal preferences even outside the choice data set, without the use of response times (provided that the assumption is valid). We now show that the use of response times makes it possible to move beyond ordinal preferences and make out-of-sample predictions of precise choice probabilities.

\begin{definition} Within a class of models, a rationalizable SCF-RT {\em predicts choice probability $\bar{p}(x,y)$} for a non-observed choice $(x,y) \in C \setminus D$ if all RUM-CFs in the class that rationalize it satisfy $G(y,x)(0)=\bar{p}(x,y)$.\end{definition}

To state the following result, for each $(x,y) \in D$ with $p(x,y)>p(y,x)$ define $\theta(x,y)$ as the $1/q'$-percentile of the response time distribution of $x$, for $q'=2p(x,y)$, i.e.,
\[F(x,y)(\theta(x,y))=\frac{1}{2 p(x,y)}.\]
Notice that $\theta(x,y)>0$ is a measure of preference intensity that displays very similar comparative statics properties to $t(x,y)$.

\begin{theorem}\label{Thm.FRUMRT} Within the class of Fechnerian RUM-CFs, a rationalizable SCF-RT predicts a choice probability for each $(x,y) \in C\setminus D$ for which there exists $z \in X$ with $(x,z), (y,z) \in D$. Assuming $p(x,z) \geq p(y,z)$ w.l.o.g., the prediction is
\[\bar{p}(x,y)= \begin{cases}
p(x,z) F(x,z)(\theta(y,z)) & \text{if } p(y,z)>1/2,\\
p(x,z) & \text{if } p(y,z)=1/2,\\
1 - p(z,x)F(z,x)(\theta(z,y)) & \text{if } p(y,z)<1/2.
\end{cases}\]
\end{theorem}

\begin{proof} Let $(u,g,r)$ be any Fechnerian RUM-CF which rationalizes an SCF-RT $(p,f)$. For any fixed $(x,y) \in C$, this particular RUM-CF predicts
\begin{align}\label{pred1}p(x,y)=G(y,x)(0)=G(v(x,y)).\end{align}
Let $(x,y) \in C \setminus D$ and $z \in X$ such that $(x,z), (y,z)\in D$ and w.l.o.g.\ $p(x,z)\geq p(y,z)$. We distinguish three cases.

\textit{Case 1: $p(y,z)>1/2$.} It follows from Proposition \ref{Prop.SRUM} that $v(y,z)>0$. From (\ref{ruratsc}) we obtain
\[p(y,z)F(y,z)(t)=1 - G(y,z)(r^{-1}(t))\]
for all $t > 0$, and hence, by (RUM.2) and the Fechnerian assumption,
\[p(y,z)F(y,z)(t)=G(z,y)(-r^{-1}(t))=G(-r^{-1}(t)-v(z,y))=G(v(y,z)-r^{-1}(t)).\]
Evaluating this equality at $t=r(v(y,z))$ yields
\[F(y,z)(r(v(y,z)))=\frac{G(0)}{p(y,z)}=\frac{1}{2p(y,z)}\]
(recall that $G(0)=1/2$). Hence, by definition of $\theta(y,z)$ we have $\theta(y,z)=r(v(y,z))$. Analogously, we also obtain
\[p(x,z)F(x,z)(t)=G(v(x,z)-r^{-1}(t))\]
for all $t > 0$, which for $t=r(v(y,z))$ yields 
\[p(x,z)F(x,z)(r(v(y,z)))=G(v(x,z)-v(y,z))=G(v(x,y)).\]
Combined with (\ref{pred1}) and the above expression for $\theta(y,z)$, this implies
\[p(x,y)=p(x,z)F(x,z)(\theta(y,z)),\]
which is the model-independent prediction $\bar{p}(x,y)$ given in the statement.

\textit{Case 2: $p(y,z)=1/2$.} It follows from Proposition \ref{Prop.SRUM} that $v(y,z)=0$. We obtain 
\[p(x,y)=G(v(x,y))=G(v(x,z)-v(y,z))=G(v(x,z))=p(x,z),\]
which is the model-independent prediction $\bar{p}(x,y)$ given in the statement.

\textit{Case 3: $p(y,z)<1/2$.} It follows from Proposition \ref{Prop.SRUM} that $v(z,y)>0$. Following the same steps as in Case 1 but with reversed order for each pair of alternatives yields the model-independent prediction 
\[\bar{p}(x,y)=1-\bar{p}(y,x)=1 - p(z,x)F(z,x)(\theta(z,y)),\]
as given in the statement.\end{proof}

To understand the probability formula in the theorem, just consider the case where $p(x,z) > p(y,z) > 1/2$. Then $u(x) > u(y) > u(z)$ must hold under the Fechnerian assumption, where the first inequality follows from Proposition \ref{Prop.FRUM} and the second inequality follows from Proposition \ref{Prop.SRUM}. Hence we can conclude that the unknown $p(x,y)$ must be strictly smaller than $p(x,z)$, because Fechnerian choice probabilities are strictly monotone in the underlying utility differences $v(\cdot,\cdot)$ across binary choice problems. The theorem now shows that a prediction for $p(x,y)$ can be obtained by multiplying the observed $p(x,z)$ with a discounting factor $F(x,z)(\theta(y,z))$. This factor is an observable, response-time-based indicator of the relative position of $u(y)$ within the interval $[u(x),u(z)]$.

Combined with the Fechnerian assumption, the use of response times allows to predict exact choice probabilities out-of-sample. Without response times, making such a prediction would require assuming a complete and specific functional form for the utility distribution. Hence, in analogy to our earlier results, response times again serve as a substitute for stronger distributional assumptions.\footnote{For the analyst who is willing to make the strong distributional assumptions required, for example, by probit or logit models, response times have no additional value when the available data on choice is rich. However, the literature has shown that the use of response times can be valuable even in the context of logit or probit models when choice data is scarce \citep[e.g.][]{Clithero18,KonovalovKrajbich17}. Our paper differs from these studies by showing that response time data enable the recovery of preferences even when rich choice data cannot recover them, that is, in the absence of untestable assumptions on utility noise.}

\section{Behavioral Models from Economics and Psychology}\label{Sec.Specific}

Theorem \ref{Thm.RUMRT} provided a sufficient condition for preference revelation without distributional assumptions. This analysis left open two questions. First, which SCF-RTs are rationalizable by RUM-CFs? And second, how tight is the sufficient condition? In this section, we try to answer these questions by studying SCF-RTs which are generated by specific behavioral models from the literature. We will do this first for standard RUMs from economics and microeconometrics to which we add chronometric functions, and second for standard sequential sampling models from psychology and neuroscience.

\subsection{The View from Economics}

Specific RUMs are commonly used for microeconometric estimation using choice data. These models typically start from a utility function $u: X \rightarrow \mathbb{R}$ and add an error term with zero mean to each option, such that the overall utility of $x \in X$ is given by a random variable \[\tilde{u}(x) = u(x) + \tilde{\epsilon}.\] As a next step, even more specific distributional assumptions are imposed. Two popular examples are the probit and the logit model. In the probit model, the error $\tilde{\epsilon}$ is assumed to be normally distributed and i.i.d.\ across the options. The distribution of the random utility difference $\tilde{v}(x,y)=\tilde{u}(x) - \tilde{u}(y)$ is then also normal and can be described by \[G(x,y)(v)= \Phi\left(\frac{v-v(x,y)}\sigma\right),\] where $\sigma$ is a standard deviation parameter and $\Phi$ is the cdf of the standard normal distribution. This simple specification gives rise to a Fechnerian model. Generalizations allow for heteroscedasticity or correlation between the error terms of different options, in which case the parameter $\sigma$ becomes choice-set-dependent, written $\sigma(x,y)$. Such a generalized model is no longer Fechnerian but still symmetric. In the logit model, the error $\tilde{\epsilon}$ is assumed to follow a Gumbel distribution, again i.i.d.\ across the options. In that case, the random utility difference follows a logistic distribution described by \[G(x,y)(v)= \left[1+e^{-\left(\frac{v-v(x,y)}{s}\right)}\right]^{-1},\] where $s$ is a scale parameter. This model is again Fechnerian, whereas one could think of generalizations where the scale parameter becomes choice-set-dependent, in which case it is no longer Fechnerian but still symmetric.

We now treat an arbitrary symmetric RUM-CF as the real data-generating process and apply our preference revelation method to the resulting data of choices and response times. It is trivial that this data is rationalizable within the class of symmetric models, and hence within the class of all models. More surprising is the fact that our sufficient condition from Theorem \ref{Thm.RUMRT} always recovers the correct preferences from the data.

\begin{proposition}\label{Prop.RatRUMs} Consider an SCF-RT $(p,f)$ that is generated by a symmetric RUM-CF $(u,g,r)$. Then, for any $(x,y) \in D$, $u(x) \geq u(y)$ implies $F(y,x)$ $q$-FSD $F(x,y)$, and $u(x) > u(y)$ implies $F(y,x)$ $q$-SFSD $F(x,y)$, for $q=p(x,y)/p(y,x)$.\end{proposition}

If we are given a data set that is generated by one of the random utility models commonly employed in the literature, augmented by a chronometric function, then our cautious revealed preference criterion always recovers the correct preferences, despite not using information about the specific data-generating process. Even the analyst who believes in the probit or logit distribution can work with our criterion. It will yield the same revealed preferences as an application of the full-fledged model if his belief is correct, but avoids a mistake if his belief is incorrect. 

We add two remarks on this result. First, symmetry of the RUM-CF does not guarantee that also the stronger condition in Corollary \ref{Cor.RUMRT1} is satisfied by the SCF-RT.\footnote{The example of an SCF-RT presented in Appendix \ref{App.Additional}, which violates this stronger condition, is actually generated by a symmetric RUM-CF. This RUM-CF features a bimodal utility difference distribution. It is possible to show that symmetry and unimodality together imply that the stronger sufficient condition in Corollary \ref{Cor.RUMRT1} is always satisfied. Both the probit and the logit model are unimodal, so working with the stronger condition comes with no loss if either of these models is the data-generating process.} Second, and more importantly, symmetry of the RUM-CF is not necessary for our condition in Theorem \ref{Thm.RUMRT} to be applicable. As the proof of Proposition \ref{Prop.RatRUMs} reveals, there are also many asymmetric models that generate data from which our criterion recovers preferences correctly.\footnote{\label{Fn.Asym}Formally, our proof relies on the property that $v(x,y)\geq 0$ implies $1-G(x,y)(v) \geq G(x,y)(-v)$ for all $v\geq0$, with strict inequality for some $v\geq 0$ when $v(x,y)>0$. This property is satisfied by symmetric models, but can also be satisfied by asymmetric models.}

We can go one step further and assume that there is a second source of noise in response times, on top of the noise already generated by random utility. The noise could be part of the behavioral model, e.g.\ due to a stochastic chronometric function or randomness in the physiological process implementing the response, or it could be due to imperfect observation by the analyst. We assume that this additional noise is independent from the randomness in utility, so that it does not systematically reverse the chronometric relationship. A common approach in the empirical literature \citep[e.g.,][]{Chabrisetal09,FischbacherHertwigBruhin13,Reinforcement3x3} is to add i.i.d.\ noise to log response times, where taking the log ensures that response times remain non-negative. An equivalent way of modelling this is by means of multiplicative noise, which is technically convenient for our purpose. Formally, a \textit{random utility model with a noisy chronometric function} (RUM-NCF) can be obtained from a RUM-CF by letting the response time for a realized utility difference $v \in \mathbb{R}$ become the random variable \[\tilde{r}(v)= r(|v|) \cdot \tilde{\eta}.\] Here, $\tilde{\eta}$ is a non-negative random term with mean one, assumed to be i.i.d.\ according to a density $h$ on $\mathbb{R}^+$. The probability of a realized response time of at most $t>0$, conditional on $x$ being chosen over $y$, is now the probability that the realized utility difference is at least $r^{-1}(t/\tilde{\eta})$, conditional on that difference being positive. Hence, for an SCF-RT $(p,f)$ that is generated by a RUM-NCF $(u,g,r,h)$ we have \[ \frac{\int_{0}^{\infty} \left[1-  G(x,y)(r^{-1}(t/\eta))\right] h(\eta) d \eta }{ 1-G(x,y)(0)}= F(x,y)(t)\] for all $t>0$ and all $(x,y) \in D$, which is analogous to equation (\ref{ruratsc}) in Definition \ref{defratscfrt}.

Our preference revelation approach based on RUM-CFs is misspecified when the real data-generating process is a RUM-NCF, because the additional noise is erroneously explained by additional randomness in utility. However, as the next proposition will show, this misspecification is often inconsequential. For the entire class of SCF-RTs that are generated by (and hence are rationalizable in) the class of symmetric RUM-NCFs with full support utility distributions $g$ (like probit or logit) and arbitrary noise distributions $h$, our previous condition remains the correct criterion for preference revelation.\footnote{Without full support of the utility difference distributions, some response times may arise only because of the additional noise but could never be generated by a realized utility difference. The distribution of those response times would be uninformative of utility and does not obey the chronometric relationship.}

\begin{proposition}\label{Prop.RatNRUMs} Consider an SCF-RT $(p,f)$ that is generated by a symmetric RUM-NCF $(u,g,r,h)$ where each $g(x,y)$ is strictly positive. Then, for any $(x,y) \in D$, $u(x) \geq u(y)$ implies $F(y,x)$ $q$-FSD $F(x,y)$, and $u(x) > u(y)$ implies $F(y,x)$ $q$-SFSD $F(x,y)$, for $q=p(x,y)/p(y,x)$.\end{proposition}

The proof rests on the insight that $q$-FSD is invariant to independent perturbations. Whenever an SCF-RT $(p,f)$ satisfies $F(y,x)$ $q$-FSD $F(x,y)$, then the SCF-RT $(p,\hat{f})$ obtained after perturbing response times by log-additive or multiplicative noise still satisfies $\hat{F}(y,x)$ $q$-FSD $\hat{F}(x,y)$. The case of RUM-NCFs obtained from symmetric RUM-CFs is a natural application of this insight. However, the robustness of our preference revelation criterion holds more generally for perturbations of any data-generating process for which the criterion has bite. This could be a random utility model that is not symmetric, or it could be one of the sequential sampling models studied in the next section.


\subsection{The View from Psychology}

A different way of generating stochastic choices and response times is by means of a sequential sampling model as used extensively in psychology and neuroscience. The basic building block for binary choice problems is the {\em drift-diffusion model} (DDM) of \citet{Ratcliff78}. A \textit{DDM with constant boundaries} is given by a drift rate $\mu \in \mathbb{R}$, a diffusion coefficient $\sigma^2>0$, and symmetric barriers $B$ and $-B$ with $B>0$. A stochastic process starts at $Z(0)=0$ and evolves over time according to a Brownian motion
\[dZ(t)=\mu dt + \sigma dW(t).\]
The process leads to a choice of $x$ (resp.\ $y$) if the upper (resp.\ lower) barrier is hit first, the response time being the time at which this event occurs. 

Although the DDM, as a model anchored in psychology, usually does not make reference to underlying utilities, $Z(t)$ is often interpreted as the difference in spiking rates between neurons computing values for the competing options. Hence it is natural to introduce a link to utility by assuming that the drift rate is determined such that $\mu=\mu(x,y)=-\mu(y,x) \geq 0$ if and only if $v(x,y)\geq 0$. This way, the DDM generates stochastic choices and response times from an underlying deterministic utility function $u: X \rightarrow \RR.$ 

The stochastic path of $Z(t)$ can be interpreted as the accumulation of evidence in favor of one or the other option as the brain samples past (episodic) information. Recent research has shown that evidence-accumulation models like the DDM actually represent optimal decision-making procedures under neurologically founded constraints, but optimality requires that the barriers are not constant but rather collapse towards zero as $t$ grows to infinity \citep{TajimaDrugowitschPouget16}. Similarly, \citet{FudenbergStrackStrzalecki18} model optimal sequential sampling when utilities are uncertain and gathering information is costly, and find that the range for which the agent continues to sample should collapse to zero as $t$ grows. A partial intuition for this result is that a value of $Z(t)$ close to zero for small $t$ carries little information, while a value of $Z(t)$ close to zero for large $t$ indicates that the true utilities are most likely close to each other, and hence sampling further evidence has little value. To reflect this idea, a {\em DDM with collapsing boundaries} works in the exact same way as a DDM with constant boundaries, with the only difference that the barriers are given by a continuous and strictly decreasing function $b:\RR^+\to\RR^{++}$ such that $\lim_{t\to \infty} b(t) = 0$. That is, $x$ is chosen if $Z(t)$ hits the upper barrier $b(t)$ before hitting the lower barrier $-b(t)$, and $y$ is chosen if the converse happens, with the response time being the first crossing time  (see Figure \ref{Fig.DDM}).

We now treat a DDM with an underlying utility function as the real data-generating process and again apply our preference revelation method to the resulting data of choices and response times.\footnote{For the case of constant boundaries, closed-form solutions for choice probabilities and response time distributions generated by the DDM are known, see e.g.\ \citet{Palmeretal05}. Closed-form solutions are not available for the case of collapsing boundaries. \citet{Webb18} explores the link between bounded accumulation models as the DDM and random utility models and shows how to derive distributional assumptions for realized utilities of the latter if the true data-generating process is of the DDM form.} The following proposition, the proof of which relies on a result by \citet{FudenbergStrackStrzalecki18}, shows that our sufficient condition from Theorem \ref{Thm.RUMRT} is again tight and always recovers the correct preferences.

\begin{proposition}\label{Prop.RatDDMs} Consider an SCF-RT $(p,f)$ that is generated by a DDM with constant or collapsing boundaries and underlying utility function $u$. Then, $u(x) \geq u(y)$ implies $F(y,x)$ $q$-FSD $F(x,y)$, and $u(x) > u(y)$ implies $F(y,x)$ $q$-SFSD $F(x,y)$, for $q=p(x,y)/p(y,x)$.\end{proposition}

To understand the intuition behind the proof, recall from our discussion of Theorem \ref{Thm.RUMRT} that the condition ensuring a revealed preference can be reformulated as $P(x,y)(t)\geq P(y,x)(t)$ for all $t$, that is, the probability of choosing $x$ over $y$ before any pre-specified response time $t$ should be larger than the probability of choosing $y$ over $x$ before $t$. In a DDM, $u(x) > u(y)$ means that the drift rate $\mu(x,y)>0$ favors $x$ (since the upper barrier reflects a choice of $x$). Hence the probability of hitting the upper barrier first before any pre-specified response time $t$ is indeed larger than the probability of hitting the lower barrier first. The proof actually establishes that the stronger condition in Corollary \ref{Cor.RUMRT1} is always satisfied by an SCF-RT that is generated by a DDM. Furthermore, the result would continue to hold for DDMs with more general boundary functions that are not necessarily constant or collapsing, but these models do not always generate well-behaved choices and have received less attention in the literature.

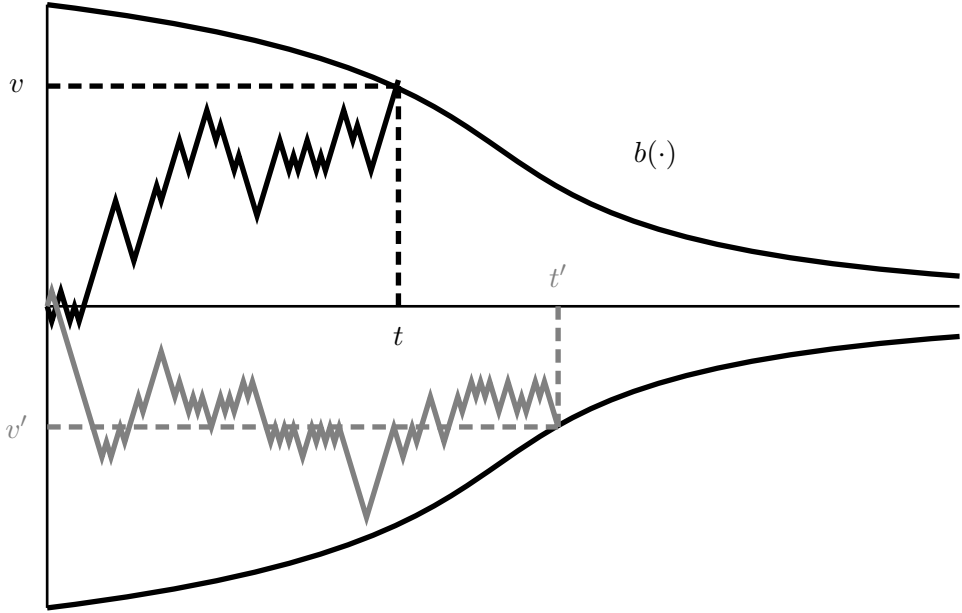
\begin{figure}[!t]
\centering
\psset{xunit=2cm,yunit=4cm,linewidth=2pt}
\begin{pspicture}(0,-1)(6,1)
\psline[linewidth=1pt](0,0)(6,0)
\psline[linewidth=1pt](0,-1)(0,1)
\psbezier(0,1)(4,0.75)(2,0.25)(6,0.1)
\psbezier(0,-1)(4,-0.75)(2,-0.25)(6,-0.1)

\psline(0,0)(0.03,-0.05)(0.06,0)(0.09,0.05)(0.12,0)(0.15,-0.05)(0.18,0)(0.21,-0.05)(0.24,0)(0.27,0.05)(0.3,0.1)(0.33,0.15)(0.36,0.2)(0.39,0.25)(0.42,0.3)(0.45,0.35)
(0.48,0.3)(0.51,0.25)(0.54,0.2)(0.57,0.15)(0.6,0.2)(0.63,0.25)(0.66,0.3)(0.69,0.35)
(0.72,0.4)(0.75,0.35)(0.78,0.4)(0.81,0.45)(0.84,0.5)(0.87,0.55)(0.9,0.5)(0.93,0.45)
(0.96,0.5)(0.99,0.55)(1.02,0.6)(1.05,0.65)(1.08,0.6)(1.11,0.55)(1.14,0.6)(1.17,0.55)
(1.2,0.5)(1.23,0.45)(1.26,0.5)(1.29,0.45)(1.32,0.4)(1.35,0.35)(1.38,0.3)(1.41,0.35)
(1.44,0.4)(1.47,0.45)(1.5,0.5)(1.53,0.55)(1.56,0.5)(1.59,0.45)(1.62,0.5)(1.65,0.45)
(1.68,0.5)(1.71,0.55)(1.74,0.5)(1.77,0.45)(1.8,0.5)(1.83,0.45)(1.86,0.5)(1.89,0.55)
(1.92,0.6)(1.95,0.65)(1.98,0.6)(2.01,0.55)(2.04,0.6)(2.07,0.55)(2.1,0.5)(2.13,0.45)
(2.16,0.5)(2.19,0.55)(2.22,0.6)(2.25,0.65)(2.28,0.7)(2.31,0.75)

\psline[linecolor=gray](0,0)(0.03,0.05)(0.06,0)(0.09,-0.05)(0.12,-0.1)(0.15,-0.15)
(0.18,-0.2)(0.21,-0.25)(0.24,-0.3)(0.27,-0.35)(0.3,-0.4)(0.33,-0.45)(0.36,-0.5)(0.39,-0.45)
(0.42,-0.5)(0.45,-0.45)(0.48,-0.4)(0.51,-0.45)(0.54,-0.4)(0.57,-0.35)(0.6,-0.3)(0.63,-0.35)
(0.66,-0.3)(0.69,-0.25)(0.72,-0.2)(0.75,-0.15)(0.78,-0.2)(0.81,-0.25)(0.84,-0.3)(0.87,-0.25)
(0.9,-0.3)(0.93,-0.35)(0.96,-0.3)(0.99,-0.35)(1.02,-0.3)(1.05,-0.35)(1.08,-0.4)(1.11,-0.35)
(1.14,-0.3)(1.17,-0.35)(1.2,-0.3)(1.23,-0.35)(1.26,-0.3)(1.29,-0.25)(1.32,-0.3)(1.35,-0.25)
(1.38,-0.3)(1.41,-0.35)(1.44,-0.4)(1.47,-0.45)(1.5,-0.4)(1.53,-0.45)(1.56,-0.4)(1.59,-0.45)
(1.62,-0.4)(1.65,-0.45)(1.68,-0.5)(1.71,-0.45)(1.74,-0.4)(1.77,-0.45)(1.8,-0.4)(1.83,-0.45)
(1.86,-0.4)(1.89,-0.45)(1.92,-0.4)(1.95,-0.45)(1.98,-0.5)(2.01,-0.55)(2.04,-0.6)(2.07,-0.65)
(2.1,-0.7)(2.13,-0.65)(2.16,-0.6)(2.19,-0.55)(2.22,-0.5)(2.25,-0.45)(2.28,-0.4)(2.31,-0.45)
(2.34,-0.5)(2.37,-0.45)(2.4,-0.4)(2.43,-0.45)(2.46,-0.4)(2.49,-0.35)(2.52,-0.3)(2.55,-0.35)
(2.58,-0.4)(2.61,-0.45)(2.64,-0.4)(2.67,-0.35)(2.7,-0.3)(2.73,-0.35)(2.76,-0.3)(2.79,-0.25)
(2.82,-0.3)(2.85,-0.25)(2.88,-0.3)(2.91,-0.25)(2.94,-0.3)(2.97,-0.35)(3,-0.3)(3.03,-0.25)
(3.06,-0.3)(3.09,-0.35)(3.12,-0.3)(3.15,-0.35)(3.18,-0.3)(3.21,-0.25)(3.24,-0.3)(3.27,-0.25)
(3.3,-0.3)(3.33,-0.35)(3.36,-0.4)

\psline[linestyle=dashed](2.31,0)(2.31,0.73)
\psline[linestyle=dashed](0,0.73)(2.31,0.73)

\rput[c](2.31,-0.1){$t$}
\rput[c](-0.2,0.73){$v$}
\rput[c](4,0.5){$b(\cdot)$}

\psline[linecolor=gray,linestyle=dashed](3.36,0)(3.36,-0.4)
\psline[linecolor=gray,linestyle=dashed](0,-0.4)(3.36,-0.4)
\rput[c](3.36,0.1){\gray $t'$}
\rput[c](-0.2,-0.4){\gray $v'$}

\end{pspicture}
\caption{Illustration of a DDM with collapsing decision boundaries.}\label{Fig.DDM}
\end{figure}

We conclude the section with remarks about rationalizability. An SCF-RT that is generated by a DDM is always rationalizable within the class of all RUM-CFs when we consider just two options, $D=\{(x,y),(y,x)\}$, which is the setting in which DDMs are typically applied. This is a corollary of the following more general result.

\begin{proposition}\label{Prop.Rat2Opt} Suppose $D=\{(x,y),(y,x)\}$. Then, any SCF-RT $(p,f)$ is rationalizable within the class of all RUM-CFs.\end{proposition}

In the proof, we fix an arbitrary chronometric function $r$ that satisfies Definition \ref{defrumrt} (and has $r(v)=0$ for large $v$) and construct an associated density $g(x,y)$ such that the data is rationalized. The construction is particularly illuminating for an SCF-RT that is generated by a DDM with collapsing boundaries. In that case, it is very natural to choose the chronometric function of the RUM-CF as the inverse of the boundary function of the DDM. The associated density $g(x,y)$ then describes the distribution of the value of $Z(t)$ at the endogenous decision point (see again Figure \ref{Fig.DDM}). This interpretation is in line with \citet{FudenbergStrackStrzalecki18}, who argue that $Z(t)$ is a signal about the true utility difference. In other words, if we think of the chronometric function as the inverse of the collapsing boundary function, then realized utility differences can be interpreted as realized signals about the underlying deterministic utility difference.

With more than two options, the question of rationalizability is less straightforward. It is not clear that the above construction yields utility difference distributions that are consistent with a single utility function $u$. For SCF-RTs generated by DDMs, this problem is further complicated by the fact that there is no agreed upon discipline on how utility differences $v(x,y)$ map into drift rates $\mu(x,y)$, beyond the basic ordinal requirement that $\mu(x,y)\geq 0$ if and only if $v(x,y)\geq 0$. Hence, rationalizability of an arbitrary DDM-generated data set in terms of a RUM-CF is not guaranteed. Our approach, however, can be generalized naturally to solve this problem. Define a {\em generalized random utility model with a chronometric function} (GRUM-CF) by replacing (RUM.1) in Definition \ref{def.rum} by the weaker requirement \begin{enumerate}[(GRUM.1)]
\item[(GRUM.1)] $\int_{-\infty}^{+\infty} v g(x,y)(v) d v\geq 0 \; \Leftrightarrow \; u(x) - u(y) \geq 0$,
\end{enumerate}
keeping everything else unchanged. That is, realized utility differences (or ``decision values'' as in the recent neuroeconomics literature, see e.g.\ \citealp[][Chapters 8--10]{Neuroecon2ndEd}) are interpreted as signals  about ordinal preferences rather than cardinal utility. It is easy to see that Theorem \ref{Thm.RUMRT} remains valid for preference revelation within the class of all GRUM-CFs. Furthermore, any SCF-RT generated by a DDM with constant or collapsing boundaries and underlying utility function $u$ is rationalizable in the class of all GRUM-CFs.\footnote{This approach could be extended even further, building a connection to the economic literature on consumer theory without transitive preferences. \citet{Shafer74} showed that every complete, continuous, and strongly convex binary relation $R$ (not necessarily transitive) on a Euclidean space can be represented by a continuous, real-valued, two-variable function $v$ such that $xRy$ if and only if $v(x,y)\geq 0$, with $v(x,y)=-v(y,x)$. Hence, one could replace $u$ in the definition of RUM-CFs and DDMs by a complete but not necessarily transitive binary relation $R$ on $X$. The appropriate reformulation of (RUM.1) would be $\int_{-\infty}^{+\infty} v g(x,y)(v) d v\geq 0 \; \Leftrightarrow \; x R y$, and DDMs could be linked to $R$ by $\mu(x,y)\geq 0 \; \Leftrightarrow \; x R y$.}

\section{Relation to the Literature}\label{Sec.Lit}

\subsection{Response Times}\label{Subsec.LitRT}

Our work is related to a recent strand of the literature which makes the empirical point that response time data can help with structural estimation of preferences by using the chronometric function. \citet{SchotterTrevino14} and \citet{KonovalovKrajbich17} propose to estimate indifference points by using the longest response times in a data set and then deduce ordinal preference relations from the indifference points. Other studies have shown how response times are indicative of effort allocation \citep{Moffatt05} or can be used to improve out-of-sample predictions of choices \citep{Clithero18}. All those works, however, consider fully-specified structural models and add response times to improve the estimation, an approach which can be useful when choice data is scarce or not reliable. In contrast, our work shows that response times enable important improvements in the recovery of preferences even when rich choice data are available. We provide a simple and intuitive condition on response times which ensures that preferences can be recovered in the absence of any assumptions on the distribution of utility noise, i.e., under conditions where the recovery of preferences fails even with rich choice data. For data sets generated by models from a large class of data-generating processes, our condition is not only sufficient but also necessary for recovering preferences (in the sense of Propositions \ref{Prop.RatRUMs} - \ref{Prop.RatDDMs}). 

\citet{EcheniqueSaito17} provide an axiomatization of the chronometric relationship, viewed as a mapping from utility differences to response times as in our model. They consider deterministic choices and deterministic response times only. Their main interest is a characterization of finite and incomplete data sets that can be rationalized by a deterministic utility function together with a chronometric function. That is, they do not consider stochastic choice or the problems that arise when utility is noisy.

While response times are generally receiving increased attention as a tool to improve economic analysis, detailed studies are still scarce (a review and discussion can be found in \citealp{SpiliopoulosOrtmann18}). Examples include the studies of risky decision-making by \citet{Wilcox93,Wilcox94}, the web-based studies of \citet{RubinsteinRT07,Rubinstein13}, and the study of belief updating by \citet{BayesRT}. The study of \citet{PrefRev} uses the chronometric relationship to understand the preference reversal phenomenon \citep{GretherPlott79}, where decision-makers typically make lottery choices which contradict their elicited certainty equivalents when one of the lotteries has a salient, large outcome. \citet{PrefRev} show that, if reversals are due to a bias in the elicitation process rather than in the choice process, choices associated with reversals should take longer than comparable non-reversals. The reason is simply that reversals (where noisy valuations ``flip'') are more likely when the actual utilities are close, and hence, by the chronometric relationship, response times must be longer. The prediction is readily found in the data, providing insights into the origin and nature of reversals.

\subsection{Stochastic Choice}\label{Subsec.LitSC}

Our work is also related to the recent literature on stochastic choice theory using extended data sets. \citet{CaplinMartin15} and \citet{CaplinDean15} consider state-dependent data sets, which specify choice frequencies as functions of observable states. \citet{CaplinMartin15} study rationalizability by maximization of expected utility when the decision-maker has a prior on the state and updates it through Bayes' rule after receiving signals on the state. \citet{CaplinDean15} study rationalizability when the decision-maker additionally decides how much effort to invest in obtaining costly signals (through attention strategies). While their results focus on rationalizability, state-dependent choice data adds an additional dimension which potentially could help with preference revelation. 

Our paper also contributes to the theory of behavioral welfare economics, which aims at eliciting preferences from inconsistent choice data \citep[among others, see][]{BernheimRangel09,RubinsteinSalant12,Masatlioglu12,BenkertNetzer18}. Most of this literature considers deterministic choices. Two exceptions are \citet{ManziniMariotti14} and \citet{ApesteguiaBallester15}. \citet{ManziniMariotti14} show that underlying preferences can be identified when stochastic choice is due to stochastic consideration sets. \citet{ApesteguiaBallester15} propose as a welfare measure the preference relation which is closest (in a certain, well-defined sense) to the observed stochastic choices. Similarly to our results in Section \ref{Sec.Specific}, they show that this procedure recovers the true underlying preference if the data is generated by random utility models fulfilling a monotonicity condition. To the best of our knowledge, this literature has not yet discovered the value of response time data for preference revelation.

The difficulty of a choice problem can be influenced by additional factors, on top of the utility difference between the options. For instance, if the options are multidimensional, then a choice problem involving a dominant alternative may be very simple, generating accurate and quick responses even if the underlying utility difference is small. The same applies to different framings of the same problem \citep[see e.g.][]{SalantRubinstein08,BenkertNetzer18}. This is a well-known empirical problem of conventional RUMs, independently of whether they are enriched with response times or not. However, our definition of RUMs sidesteps this problem. The reason is that the added generality in Definition \ref{def.rum} allows for pair-specific pdfs $g(x,y)$. Hence, if $u(x)-u(y)$ and $u(x')-u(y')$ are similar but the pairs $(x,y)$ and $(x',y')$ vary along a dimension not captured by underlying utilities alone, the pdfs $g(x,y)$ and $g(x',y')$ can differ reflecting this additional dimension. For instance, \citet{HeNatenzon18} consider ``moderate utility models'' relating choice probabilities $p(x,y)$ to utility differences $u(x)-u(y)$ and an additional distance $d(x,y)$ which reflects choice difficulty beyond utility differences. Specifically, the assumption is that $p(w,x)\geq p(y,z)$ if and only if $(u(w)-u(x))/d(w,x) \geq (u(y)-u(z))/d(y,z)$, which retains the basic regularities of conventional models while allowing for violations of weak stochastic transitivity. This is also encompassed by our definition of RUMs with pair-specific utility difference distributions. A similar point applies if decisions are subject to pair-specific impulsive tendencies or intuitive processes which might influence response times independently of the underlying utilities \citep{BayesRT}.

\section{Conclusion}\label{Sec.Conclusion}

Choice theory has traditionally focused on choice outcomes and has ignored auxiliary data such as response times. This neglect comes at a cost even for traditional choice-theoretic questions. In the context of stochastic choice, ignoring response time data means discarding information about the distribution of random utility, which then has to be compensated by making distributional assumptions. In this paper, we have developed a suite of tools to utilize response time data for a recovery of preferences without or with fewer distributional assumptions.

Throughout most of the paper, we have interpreted SCF-RTs as describing the choices of a single individual who is confronted with the same set of options repeatedly. Random utility then reflects fluctuating tastes or noisy perception of the options. However, our tools also work when the data is generated by a heterogeneous population of individuals, each of whom makes a deterministic choice at a deterministic response time \citep[as in][]{EcheniqueSaito17}. Random utility then reflects a distribution of deterministic utility functions within the population, and response times vary because the difficulty of the choice problem varies with the subjective utility difference. At first glance, a one-to-one translation of our results seems to require the assumption that the same chronometric function applies to all individuals.\footnote{Empirical results by \citet{Chabrisetal09} and \citet{KonovalovKrajbich17} indicate that response times (even as little as one observation per individual) can indeed be used to track down parametric differences in utilities across individuals.} However, what we called a noisy chronometric function in Section \ref{Sec.Specific} can readily be interpreted as a distribution of chronometric functions within the population. In this population interpretation, a revealed preference for $x$ over $y$ means that utilitarian welfare with $x$ is larger than with $y$. Thus, the use of response time data is a novel way to approach the long-standing problem of how to measure the cardinal properties of utility that utilitarianism relies on \citep[see][]{Aspremont02}. The requirement that each individual's chronometric function is drawn independently from a restricted class, as described in Section \ref{Sec.Specific}, mirrors the interpersonal comparability of utility units that utilitarianism requires.

There is a range of interesting questions that we leave for future research. First, our results lend themselves to empirical testing. A first natural step is to work with experimental choice data from the lab, where response times are easy to measure. Later work could study real-world data e.g.\ from online marketplaces, where the time a consumer spends contemplating the options could be (and presumably is already) recorded. A challenge will be to differentiate response time in our sense from other concepts such as the time required to read information or to deliberate on the consequences of an action, which may have other qualitative predictions \citep[as in][]{RubinsteinRT07,Rubinstein13}.

Second, we have not attempted a full characterization of rationalizability for arbitrary SCF-RTs beyond those studied in Section \ref{Sec.Specific}. For the case without response times, characterizations are relatively simple and have been given in the literature. For instance, it can be shown that an SCF is rationalizable in our class of symmetric RUMs if and only if the binary relation $R^s$ defined in Section \ref{Sec.Reveal} has no cycles \citep[in the sense of][]{Suzumura76}. The problem is substantially more involved when response time distributions have to be rationalized, too. However, some useful necessary conditions are easy to obtain. For instance, an SCF-RT is rationalizable in the class of symmetric RUM-CFs only if both $R^s$ and $R^{srt}$ have no cycles, as otherwise there cannot exist a utility function that is consistent with the revealed preferences. Analogous conditions hold for rationalizability in the classes of all RUM-CFs and Fechnerian RUM-CFs. These simple conditions provide a specific test of our response-time-based model and allow it to be falsified by the data. 

Finally, response times are a particularly simple measure with a well-established relation to underlying preferences, but they may not be the only auxiliary data with that property. Physiological measures such as pupil dilation, blood pressure, or brain activation may also carry systematic information about preferences. It is worth exploring to what extent these measures can improve the classical revealed preference approach and should therefore be added to the economics toolbox.
 
%
%
%

\newpage
\singlespacing
\providecommand{\noopsort}[1]{}

\newpage
\singlespacing
\section*{Appendices}

\begin{appendix}
\section{Omitted Proofs}\label{App.Proofs}

\subsection{Proof of Proposition \ref{Prop.RUM}}

Consider any SCF and fix an arbitrary utility function $u: X \rightarrow \RR$. We will construct a RUM with utility function $u$ that rationalizes the SCF. Since $u$ is arbitrary, it follows that no preference between any $x$ and $y$ with $x \neq y$ is revealed.

For all $(x,y) \in C \setminus D$, choose arbitrary densities $g(x,y)$ and $g(y,x)$ so that (RUM.1-3) are satisfied. For $(x,y) \in D$, where w.l.o.g.\ $v(x,y) \geq 0$, define
\[g(x,y)(v) = \begin{cases}
0 & \text{if } \delta(x,y) < v\\
d(x,y) & \text{if } 0 \leq v \leq \delta(x,y)\\
p(y,x) & \text{if } -1 \leq v < 0\\
0 & \text{if } v < -1
\end{cases}\]
and $g(y,x)(v)=g(x,y)(-v)$ for all $v \in \mathbb{R}$, with \[d(x,y)=\frac{p(x,y)^2}{p(y,x)+2v(x,y)} > 0 \, \text{ and } \, \delta(x,y)=\frac{p(y,x)+2v(x,y)}{p(x,y)}>0.\]
We then obtain that
\begin{align*}\int_{-\infty}^{+\infty} v g(x,y)(v) dv & = \int_{-1}^0 v p(y,x) dv + \int_0^{\delta(x,y)} v d(x,y) dv \\ & = -\frac{1}{2} p(y,x) + \frac{1}{2} \delta(x,y)^2 d(x,y)  =v(x,y),\end{align*}
so that (RUM.1-3) are satisfied. It also follows that $G(x,y)(0)=p(y,x)$, so this RUM indeed rationalizes the SCF. \hfill $\qedsymbol$

\subsection{Proof of Corollary \ref{Cor.RUMRT1}}

The condition that $q(x,y)(t) \geq 1$ for almost all $t \geq 0$ can be rewritten as \[p(x,y)f(x,y)(t) \geq p(y,x)f(y,x)(t)\] for almost all $t \geq 0$. It implies \[p(x,y)F(x,y)(t)= \int_0^{t} p(x,y)f(x,y)(\tau) d\tau \geq \int_0^{t} p(y,x)f(y,x)(\tau) d\tau = p(y,x)F(y,x)(t)\] for all $t \geq 0$. Hence $Q(x,y)(t) \geq 1$ holds for all $t>0$, which implies a revealed preference for $x$ over $y$ by (\ref{Eq.Q}) and Theorem \ref{Thm.RUMRT}. The argument for strict preferences is analogous. \hfill $\qedsymbol$

\subsection{Proof of Corollary \ref{Cor.RUMRT2}}

Let $(u,g,r)$ be any RUM-CF which rationalizes an SCF-RT $(p,f)$. For any $x,y \in X$ with $(x,y) \in T(R^{rt})$, it follows that there exists a sequence $x_1,x_2,\ldots,x_n$ with $x_1=x$ and $x_n=y$ such that, for each $k=1,\ldots,n-1$, we have $(x_k,x_{k+1}) \in R^{rt}$ and hence $v(x_k,x_{k+1})\geq 0$ by definition of $R^{rt}$ and Theorem \ref{Thm.RUMRT} (or trivially, if $x_k=x_{k+1}$). This implies $v(x,y)=\sum_{k=1}^{n-1} v(x_k,x_{k+1}) \geq 0$, i.e., a revealed preference for $x$ over $y$.

If $(x,y) \in T_P(R^{rt})$, the above sequence cannot at the same time satisfy $(x_{k+1},x_k) \in R^{rt}$ for each $k=1,\ldots,n-1$. Hence $(x_{k^*},x_{k^*+1}) \in P^{rt}$ for some $k^*=1,\ldots,n-1$. We claim that this implies $v(x_{k^*},x_{k^*+1})>0$, and therefore $v(x,y)>0$, i.e., a revealed strict preference for $x$ over $y$. The fact that $(x_{k^*+1},x_{k^*}) \notin R^{rt}$ implies that $x_{k^*} \neq x_{k^*+1}$ and that there exists $t^* \geq 0$ such that \[F(x_{k^*},x_{k^*+1})(t^*) > \frac{p(x_{k^*+1},x_{k^*})}{p(x_{k^*},x_{k^*+1})} F(x_{k^*+1},x_{k^*})(t^*).\] Together with $(x_{k^*},x_{k^*+1}) \in R^{rt}$ this implies $F(x_{k^*+1},x_{k^*})$ $q$-SFSD $F(x_{k^*},x_{k^*+1})$ for $q=p(x_{k^*},x_{k^*+1})/p(x_{k^*+1},x_{k^*})$, so that our claim follows from Theorem \ref{Thm.RUMRT}. \hfill $\qedsymbol$

\subsection{Proof of Proposition \ref{Prop.SRUM}}

Let $(u,g)$ be a symmetric RUM which rationalizes an SCF $p$. By symmetry, we have that $G(x,y)(v(x,y))=1/2$ for all $(x,y) \in C$. Hence $G(x,y)(0) < 1/2$ implies $v(x,y) > 0$, because $G(x,y)(v)$ is increasing in $v$. Furthermore, $G(x,y)(0) = 1/2$ implies $v(x,y) = 0$, because $G(x,y)(v)$ is strictly increasing in $v$ in the connected support of $g(x,y)$, by (RUM.3). 

Suppose $p(x,y) \geq p(y,x)$ for some $(x,y)\in D$. From Definition \ref{defratrum} it then follows that $G(x,y)(0)=p(y,x)\leq 1/2$ and hence $v(x,y)\geq 0$, i.e., a revealed preference for $x$ over $y$. If $p(x,y) > p(y,x)$, then analogously $G(x,y)(0)=p(y,x)< 1/2$ and hence $v(x,y)>0$, i.e., a revealed strict preference for $x$ over $y$.\hfill $\qedsymbol$

\subsection{Proof of Corollary \ref{Cor.SRUM}}

The proof is similar to the proof of Corollary \ref{Cor.RUMRT2} and therefore omitted.\hfill $\qedsymbol$

\subsection{Proof of Corollary \ref{Cor.SRUMRT}}

Let $(u,g,r)$ be any symmetric RUM-CF which rationalizes an SCF-RT $(p,f)$. For any $x,y \in X$ with $(x,y) \in T(R^s \cup R^{srt})$, it follows that there exists a sequence $x_1,x_2,\ldots,x_n$ with $x_1=x$ and $x_n=y$ such that, for each $k=1,\ldots,n-1$, either $(x_k,x_{k+1}) \in R^{s}$ or $(x_k,x_{k+1}) \in R^{srt}$. It follows by definition of $R^{s}$ and Proposition \ref{Prop.SRUM}, or by definition of $R^{srt}$ and Theorem \ref{Thm.SRUMRT}, that $v(x_k,x_{k+1})\geq 0$ in either case. This implies $v(x,y)=\sum_{k=1}^{n-1} v(x_k,x_{k+1}) \geq 0$, i.e., a revealed preference for $x$ over $y$.

If $(x,y) \in T_P(R^s \cup R^{srt})$, there must exist $k^*=1,\ldots,n-1$ such that, in the above sequence, neither $(x_{k^*+1},x_{k^*}) \in R^s$ nor $(x_{k^*+1},x_{k^*}) \in R^{srt}$, hence either $(x_{k^*},x_{k^*+1}) \in P^s $ or $(x_{k^*},x_{k^*+1}) \in P^{srt}$. If $(x_{k^*},x_{k^*+1}) \in P^s $, then $v(x_{k^*},x_{k^*+1})> 0$ by definition of $R^s$ and Proposition \ref{Prop.SRUM}. If $(x_{k^*},x_{k^*+1}) \in P^{srt} $, then $(x_{k^*},x_{k^*+1}) \in C \setminus D$ and $\exists z \in X$ such that $t(x_{k^*},z) < t(x_{k^*+1},z)$ or $t(z,x_{k^*}) > t(z,x_{k^*+1})$ by definition of $R^{srt}$, hence $v(x_{k^*},x_{k^*+1})> 0$ by Theorem \ref{Thm.SRUMRT}. This implies $v(x,y)>0$ in either case, i.e., a revealed strict preference for $x$ over $y$.\hfill $\qedsymbol$

\subsection{Proof of Proposition \ref{Prop.FRUM}}

Let $(u,g)$ be any Fechnerian RUM which rationalizes an SCF $p$. For each $(x,y) \in C$, the Fechnerian assumption implies 
\[G(x,y)(0) = \int_{-\infty}^{0} g(x,y)(v)dv = \int_{-\infty}^{0} g(v-v(x,y)) dv=\int_{-\infty}^{-v(x,y)}g(v)dv=G(v(y,x)),\]
where $G$ is the (strictly increasing) cumulative distribution function for $g$. Hence $G(z,x)(0) \geq G(z,y)(0)$ implies $G(v(x,z)) \geq G(v(y,z))$ and therefore $v(x,z) \geq v(y,z)$, which in turn implies $v(x,y)=v(x,z)-v(y,z)\geq0$. Furthermore, $G(z,x)(0) > G(z,y)(0)$ implies $v(x,y)>0$.

Consider any $(x,y) \in C \setminus D$ such that there exists $z \in X$ with $p(x,z) \geq p(y,z)$. Hence $G(z,x)(0) \geq G(z,y)(0)$ and $v(x,y)\geq 0$, i.e., a revealed preference for $x$ over $y$. If $p(x,z) > p(y,z)$ then analogously $v(x,y) > 0$, i.e., a revealed strict preference for $x$ over $y$.\hfill $\qedsymbol$

\subsection{Proof of Corollary \ref{Cor.FRUM}}

The proof is similar to the proof of Corollary \ref{Cor.SRUMRT} and therefore omitted.\hfill $\qedsymbol$

\subsection{Proof of Proposition \ref{Prop.RatRUMs}}

Consider any symmetric RUM-CF $(u,g,r)$. Suppose $v(x,y) \geq 0$. Fix any $t>0$. If $r^{-1}(t) \geq v(x,y)$, let $\delta(t)=r^{-1}(t) - v(x,y) \geq 0$. We obtain
\begin{align*}
1 - G(x,y)(r^{-1}(t)) & = 1 - G(x,y)(v(x,y) + \delta(t)) \\
& = G(x,y)(v(x,y) - \delta(t)) \\
& = G(x,y)(-r^{-1}(t) +2 v(x,y)) \\
& \geq G(x,y)(-r^{-1}(t)),
\end{align*}
where the second equality follows from symmetry. If $r^{-1}(t) < v(x,y)$, which requires that $v(x,y)>0$, let $\delta(t)=v(x,y) - r^{-1}(t) > 0$, so that
\begin{align*}
1 - G(x,y)(r^{-1}(t)) & = 1 - G(x,y)(v(x,y) - \delta(t)) \\
& = G(x,y)(v(x,y) + \delta(t)) \\
& = G(x,y)(-r^{-1}(t) +2 v(x,y)) \\
& > G(x,y)(-r^{-1}(t)),
\end{align*}
where the second equality again follows from symmetry, and the inequality is strict since $G(x,y)(-r^{-1}(t)) < G(x,y)(v(x,y))=1/2 < G(x,y)(-r^{-1}(t) + 2v(x,y))$, because $G(x,y)(v)$ is strictly increasing in $v$ in the connected support of $g(x,y)$, by (RUM.3). 

Suppose an SCF-RT $(p,f)$ is generated by $(u,g,r)$. It follows that equation (\ref{eq.p1}) derived in the proof of Theorem \ref{Thm.RUMRT} holds for any $(x,y) \in D$ and all $t>0$. Combined with the above inequalities, whenever $v(x,y) \geq 0$ we obtain that $Q(x,y)(t) \geq 1$ for all $t>0$, or $F(y,x)$ $q$-FSD $F(x,y)$ for $q=p(x,y)/p(y,x)$. If $v(x,y) > 0$, then the above case where $r^{-1}(t) < v(x,y)$ indeed arises for large enough $t$, which additionally implies that $Q(x,y)(t) > 1$ for some $t$, or $F(y,x)$ $q$-SFSD $F(x,y)$ for $q=p(x,y)/p(y,x)$.\hfill $\qedsymbol$

\subsection{Proof of Proposition \ref{Prop.RatNRUMs}}

Let $(p,f)$ be an SCF-RT that is generated by a symmetric RUM-NCF $(u,g,r,h)$ in which each $g(x,y)$ is strictly positive on $\mathbb{R}$. Consider the underlying symmetric RUM-CF $(u,g,r)$ and note that it generates an SCF-RT $(p,\hat{f})$, i.e., response time densities $\hat{f}(x,y)$ with full support. This holds because each $g(x,y)$ is strictly positive on $\mathbb{R}$ by assumption, and $r$ must be strictly positive on $\mathbb{R}^{++}$ as otherwise $(u,g,r,h)$ would have generated an atom at the response time of zero.

For any $(x,y) \in D$, it then follows from Proposition \ref{Prop.RatRUMs} that $u(x) \geq u(y)$ implies $\hat{F}(y,x)$ $q$-FSD $\hat{F}(x,y)$ and $u(x) > u(y)$ implies $\hat{F}(y,x)$ $q$-SFSD $\hat{F}(x,y)$, for $q=p(x,y)/p(y,x)$.

Fix any $(x,y) \in D$. The fact that $(u,g,r,h)$ generates $(p,f)$ implies \begin{align}\label{ratrumncf}F(x,y)(t) = \int_{0}^{\infty} \left[ \frac{1 - G(x,y)(r^{-1}(t/\eta)) }{1 - G(x,y)(0)}\right] h(\eta) d \eta \end{align} for all $t >0$. Since $(u,g,r)$ generates $(p,\hat{f})$, we obtain \[ \frac{1 - G(x,y)(r^{-1}(z))}{1 - G(x,y)(0)} = \hat{F}(x,y)(z)\] for all $z>0$. Evaluated at $z=t/\eta$ and substituted into (\ref{ratrumncf}), this yields \[F(x,y)(t) = \int_{0}^{\infty} \hat{F}(x,y)(t/\eta) h(\eta) d \eta \] for all $t >0$, and analogously \[F(y,x)(t) = \int_{0}^{\infty} \hat{F}(y,x)(t/\eta) h(\eta) d \eta.\] Now $\hat{F}(y,x)$ $q$-FSD $\hat{F}(x,y)$ implies \[F(y,x)(t) = \int_{0}^{\infty} \hat{F}(y,x)(t/\eta) h(\eta) d \eta \leq  \int_{0}^{\infty} q \cdot \hat{F}(x,y)(t/\eta) h(\eta) d \eta = q \cdot F(x,y)(t),\] for all $t \geq 0$, i.e., $F(y,x)$ $q$-FSD $F(x,y)$. Furthermore, $\hat{F}(y,x)$ $q$-SFSD $\hat{F}(x,y)$ implies that the inequality is strict for some $t$, i.e., $F(y,x)$ $q$-SFSD $F(x,y)$.\hfill $\qedsymbol$

\subsection{Proof of Proposition \ref{Prop.RatDDMs}}

Suppose an SCF-RT $(p,f)$ is generated by a DDM with constant or collapsing boundaries and underlying utility function $u$. Consider any $(x,y)\in D$. We need to show that $v(x,y) \geq 0$ and hence $\mu(x,y)\geq 0$ implies $F(y,x)$ $q$-FSD $F(x,y)$, and $v(x,y) > 0$ and hence $\mu(x,y)> 0$ implies $F(y,x)$ $q$-SFSD $F(x,y)$, for $q=p(x,y)/p(y,x)$.

Consider first the case of constant boundaries. It is well-known \citep[see e.g.][]{Palmeretal05} that $\mu(x,y)\geq 0$ implies $p(x,y) \geq p(y,x)$ and $\mu(x,y)> 0$ implies $p(x,y) > p(y,x)$. Furthermore, the distributions of response times conditional on either choice are identical, i.e., $F(x,y)(t)=F(y,x)(t)$ for all $t \geq 0$ \citep[see again][]{Palmeretal05}. The conclusion follows then immediately.

Consider now the case of collapsing boundaries. For this case, \citet[proof of Theorem 1]{FudenbergStrackStrzalecki18} show that (adapting their notation to ours)
\[q(x,y)(t)= \frac{p(x,y) f(x,y)(t)}{p(y,x) f(y,x)(t)}= \exp\left(\frac{\mu(x,y)  b(t)}{\sigma^2/2}\right)\]
for all $t \geq 0$. We obtain that $q(x,y)(t)\geq 0$ for all $t \geq 0$ when $\mu(x,y) \geq 0$, and $q(x,y)(t)> 0$ for all $t \geq 0$ when $\mu(x,y) > 0$. The conclusion now follows as in the proof of Corollary \ref{Cor.RUMRT1}.\hfill $\qedsymbol$

\subsection{Proof of Proposition \ref{Prop.Rat2Opt}}

Suppose $D=\{(x,y),(y,x)\}$ and consider an arbitrary SCF-RT $(p,f)$. Fix any function $b:\RR^+\to\RR^{++}$ that is continuous and strictly decreasing with $\lim_{t\to \infty} b(t) = 0$. Let its inverse be $r=b^{-1}$, with the understanding that $r(v)=0$ if $v > b(0)$. Note that $r:\RR^{++}\to\RR^{+}$ satisfies the requirements of a chronometric function in Definition \ref{defrumrt}. Now define $G(x,y)$ by
\begin{align*} 
G(x,y)(v)= \begin{cases} 0 & \text{if } v < -b(0),\\ p(y,x) F(y,x)(t) & \text{if } v = -b(t) \text{ for some } t \geq 0, \\ p(y,x) & \text{if } v = 0, \\ 1 - p(x,y) F(x,y)(t) & \text{if } v = b(t) \text{ for some } t \geq 0, \\ 1 & \text{if } v > b(0),
\end{cases}
\end{align*}
which, due to the assumed properties of $b$, is well-defined and describes a distribution with connected support $[-b(0),+b(0)]$ that admits a density $g(x,y)$. Let $g(y,x)(v)=g(x,y)(-v)$ for all $v \in \mathbb{R}$. Finally, choose any function $u: X \rightarrow \mathbb{R}$ such that \[u(x)- u(y)=\int_{-b(0)}^{+b(0)} v g(x,y)(v)dv,\] and let $g(w,z)$ for all $(w,z) \in C \backslash  D$ be arbitrary so that (RUM.1-3) are satisfied. It now follows immediately that the RUM-CF $(u,g,r)$ rationalizes the SCF-RT $(p,f)$.\hfill $\qedsymbol$

\newpage

\section{Additional Material}\label{App.Additional}

\subsection{Theorem \ref{Thm.RUMRT} versus Corollary \ref{Cor.RUMRT1}}

Assume we have data on a single pair $(x,y)$ such that $p(x,y)=3/4$, $p(y,x)=1/4$,
\[f(x,y)(t)=\begin{cases}
          \frac{4}{3} t^3 & \textrm{ if } 0 \leq t < 1,\\
          \frac{4}{3}\left(\frac{2-t}{t^3}\right) & \textrm{ if } 1 \leq t < 2,\\
          \frac{4}{3}\left(\frac{t-2}{t^3}\right) & \textrm{ if } t \geq 2,
\end{cases}\]
and
\[f(y,x)(t)=\frac{4 t^3}{(1+t)^5} \;\;\; \text{ for all } t\geq 0.\]
Note that $f(x,y)$ and $f(y,x)$ are continuous densities,\footnote{The fact that $f(x,y)(t)=0$ for $t=0,2$ and $f(y,x)(t)=0$ for $t=0$ is not a problem for our SCF-RT definition which requires strictly positive densities, because only isolated points are affected.} with corresponding cdfs
\[F(x,y)(t)=\begin{cases}
          \frac{1}{3} t^4 & \textrm{ if } 0 \leq t < 1,\\
          \frac{1}{3} + \frac{4}{3}\left(\frac{t-1}{t^2}\right) & \textrm{ if } 1 \leq t < 2,\\
          1 - \frac{4}{3}\left(\frac{t-1}{t^2}\right) & \textrm{ if } t \geq 2,
\end{cases}\]
and
\[F(y,x)(t)=\frac{t^4}{(1+t )^4} \;\;\; \text{ for all } t\geq 0.\]

It is easy to see that $F(y,x)$ $q$-SFSD $F(x,y)$ holds for $q=p(x,y)/p(y,x)=3$. For $0 \leq t <1$, the condition $F(y,x)(t) \leq 3 F(x,y)(t)$ reduces to $(1+t)^4 \geq 1$, which is satisfied. For $t\geq 1$, we have $F(y,x)(t) < 1 \leq 3 F(x,y)(t)$. Hence a strict preference for $x$ over $y$ is revealed according to Theorem \ref{Thm.RUMRT}, provided that the SCF-RT is rationalizable (which we will show below). At the same time, $f(y,x)(t) > 3 f(x,y)(t)$ holds for an open interval of response times around $t=2$. This follows immediately from $f(y,x)(2) > 0 = 3 f(x,y)(2)$ and continuity of $f(x,y)$ and $f(y,x)$. Hence Corollary \ref{Cor.RUMRT1} is not applicable.

The SCF-RT is rationalizable because it is generated by the RUM-RT $(u,g,r)$ with $u(x)-u(y)=1/2$, $r(v)=1/v$, and the symmetric bimodal distribution given by
\[g(x,y)(v)=\begin{cases}
          \frac{1}{(1-v)^5} & \textrm{ if } v \leq 0,\\
          1-2v & \textrm{ if } 0 < v \leq \frac{1}{2},\\
          2v - 1 & \textrm{ if } \frac{1}{2} < v \leq 1,\\
          \frac{1}{v^5} & \textrm{ if } v > 1,
\end{cases}\]
and
\[G(x,y)(v)=\begin{cases}
          \frac{1}{4(1-v)^4} & \textrm{ if } v \leq 0,\\
          \frac{1}{4} + v(1-v) & \textrm{ if } 0 < v \leq \frac{1}{2},\\
          \frac{3}{4} - v(1-v) & \textrm{ if } \frac{1}{2} < v \leq 1,\\
          1 - \frac{1}{4v^4} & \textrm{ if } v > 1.
\end{cases}\]

\end{appendix}

\end{document}